\documentclass[prx,aps,floatfix,amsmath,superscriptaddress,twocolumn,longbibliography,nofootinbib]{revtex4-2}
\pdfoutput=1

\usepackage{amssymb,enumerate}
\usepackage{graphicx}
\usepackage{amsmath}
\usepackage{amsthm}
\usepackage{color}
\usepackage{dsfont}
\usepackage{hyperref}
\usepackage{boldline}
\usepackage{cleveref}
\usepackage{tikz}
    \usetikzlibrary{math}
    \usetikzlibrary{decorations.pathreplacing}
    \usetikzlibrary{quantikz}
\usepackage{adjustbox}
\usepackage{xcolor}
\hypersetup{
    colorlinks,
    linkcolor={red!50!black},
    citecolor={blue!50!black},
    urlcolor={blue!80!black}
}
\usepackage{xfrac}
\usepackage{empheq}
\usepackage{enumitem}
\usepackage{tabularray}

\usepackage{array}
\usepackage{graphicx,graphics,epsfig,times,bm,bbm,amssymb,amsfonts,mathrsfs}
\usepackage[normalem]{ulem}
\usepackage{subfigure}
\usepackage{dsfont}
\usepackage{braket}
\usepackage{natbib}
\usepackage{chngcntr}
\usepackage{nameref}

\newcommand{\UU}{\text{U}}

\newcommand{\V}{\mathcal{V}}
\newcommand{\SU}{\text{SU}}
\newcommand{\Sn}{\mathbb{S}_n}

\newcommand{\R}{\mathbb{R}}
\newcommand{\C}{\mathbb{C}}

\newcommand{\<}{\langle}
\renewcommand{\>}{\rangle}
\renewcommand{\H}{\mathcal{H}}
\newcommand{\Zhat}{J_z}
\newcommand{\Nhat}{a^{\dagger}a}

\newcommand{\HTC}{{H}_{\text{TC}}}
\newcommand{\UTC}{V_{\text{TC}}}
\newcommand{\UZ}{R_z}

\newcommand{\alg}{\mathfrak{alg}}

\newcommand{\jmin}{j_\text{min}}

\newcommand{\qj}{_{q,j}}
\newcommand{\Vswap}{V_{\text{qub.}\leftrightarrow\text{osc.}}}
\newcommand{\Uo}{\text{U(1)}}

\newcommand{\vac}{\ket{0}_{\text{osc}}}

\newcommand{\gTC}{g_{\text{TC}}}
\newcommand{\kosc}{\ket{k}_{\text{osc}}}

\newcommand{\1}{\mathbb{I}}
\newcommand{\qeq}{\quad=\quad}
\newcommand{\eq}{\,=\,}
\newcommand{\s}{\text{Span}}
\newcommand{\p}{\,+\,}

\newcommand{\ga}{\gamma}
\newcommand{\de}{\delta}
\newcommand{\ipo}[3]{\left\langle #1 \middle| #2 \middle| #3 \right\rangle}
\newcommand{\pure}[1]{|#1\rangle\langle #1|}

\newcommand{\bes} {\begin{subequations}}
\newcommand{\ees} {\end{subequations}}
\newcommand{\bea} {\begin{eqnarray}}
\newcommand{\eea} {\end{eqnarray}}
\newcommand{\be} {\begin{equation}}
\newcommand{\ee} {\end{equation}}

\def\al{\alpha}

\def\>{\rangle}
\def\<{\langle}

\newcommand{\ignore}[1]{}

\newtheorem{thm}{Theorem}

\newcounter{step}

\newtheorem{theorem}{Theorem}

\newtheorem{corollary}[theorem]{Corollary}

\newtheorem{lemma}[theorem]{Lemma}

\newtheorem{proposition}[theorem]{Proposition}

\crefname{tab}{Table}{Tables}
\crefname{eq}{Eq.}{Eqs.}
\crefname{section}{Sec.}{Sections}
\crefname{proposition}{Proposition}{Propositions}
\crefname{step}{Step}{Steps}

\crefformat{table}{Table #2#1#3}
\crefformat{equation}{Eq.(#2#1#3)}
\crefformat{figure}{Figure #2#1#3}
\crefformat{section}{Sec.(#2#1#3)}
\crefformat{appendix}{Appendix (#2#1#3)}
\crefformat{proposition}{Proposition #2#1#3}
\crefformat{theorem}{Theorem #2#1#3}
\crefformat{lemma}{Lemma #2#1#3}
\crefformat{corollary}{Corollary #2#1#3}
\crefformat{step}{Step #1}

\crefname{enumi}{Part}{parts}
\Crefname{enumi}{Part}{Parts}

\makeatletter 
\renewcommand\onecolumngrid{%
  \do@columngrid{one}{\@ne}%
  \def\set@footnotewidth{\onecolumngrid}%
  \def\footnoterule{\kern-6pt\hrule width 1.5in\kern6pt}%
}
\makeatother

\begin{document}

\title{Permutation-Invariant $N$-body Gates via the Tavis-Cummings Interaction}

\author{Plato Deliyannis}
\email{plato.deliyannis@duke.edu}
\affiliation{Duke Quantum Center and Department of Physics, Duke University, Durham, NC 27708, USA}
\author{Iman Marvian}
\email{iman.marvian@duke.edu}
\affiliation{Duke Quantum Center and Department of Physics, Duke University, Durham, NC 27708, USA}\affiliation{Department of Electrical and Computer Engineering, Duke University, Durham, NC 27708, USA}

\begin{abstract}
Global control provides a promising route to implementing multi-qubit gates without individual qubit addressing.
This is especially appealing for permutation-invariant (PI) gates, whose symmetry is often broken when they are compiled into individually addressed one- and two-qubit gates.
Important examples include SWAP, $\sqrt{i\text{SWAP}}$, and the $n$-qubit controlled-$Z$ gate, which is equivalent, up to two single-qubit Hadamard gates, to the multi-qubit Toffoli gate.
Motivated by this global-control perspective, we show that all PI unitaries on an arbitrary number of qubits can be realized using the Tavis-Cummings (TC) interaction, the multi-qubit version of the Jaynes-Cummings interaction, together with global uniform $z$ and $x$ fields.
Here, the $n$ qubits are identically coupled to a single bosonic mode (oscillator), which is initialized in and returned to its vacuum state.
A corollary is that all PI states, including GHZ and Dicke states, can be prepared using the same global control.
For the case $n=2$ qubits, which is particularly important in quantum computing, we also find explicit pulse sequences for implementing \textit{all} PI qubit unitaries that conserve angular momentum in the $z$ direction, using only the TC interaction and global $z$ fields.
This includes controlled-$Z$, SWAP, and $\sqrt{i\text{SWAP}}$.
\end{abstract}

\maketitle

\section{Introduction}
A central goal of quantum computing is to use experimentally accessible interactions to implement elementary operations from which complex quantum computations can be built.
Although the standard circuit model is often formulated in terms of individually controllable qubits, many physical platforms realize interactions between qubits indirectly, through their coupling to a common harmonic oscillator, or bosonic mode.
A fundamental model for resonant qubit-oscillator coupling is the Jaynes-Cummings (JC) interaction \cite{Jaynes_Cummings_1963,JC_history}, which has been widely used, together with variants thereof, to realize multi-qubit operations for universal quantum computation \cite{Childs_Chuang_2000,Yuan_Lloyd_2007,Cirac_Zoller_1995}.
Prominent examples include trapped ions coupled to collective vibrational modes \cite{Leibfried_2003_trappedion,Haffner_2008_trappedion}, superconducting qubits coupled to microwave resonators in circuit QED \cite{Blais_2004_cQED,Wallraff_2004_cQED}, and atoms coupled to cavity modes in cavity QED \cite{Raimond_2001_cavity,Walther_2006_cavity,Varcoe_2000_cavity}.
This bosonic-mode-mediated paradigm has played a central role in the implementation of multi-qubit gates, including the M{\o}lmer-S{\o}rensen entangling gate in trapped-ion systems \cite{Molmer_Sorensen}.

In the usual approach, a desired multi-qubit unitary is decomposed into a sequence of elementary one- and two-qubit gates, each of which is then implemented by addressing the relevant qubits individually.
This paradigm places substantial demands on control: even when the target operation has a high degree of symmetry, its standard implementation often breaks that symmetry at the level of the control sequence.
For example, the SWAP gate, which exchanges the states of two qubits and is itself permutation-invariant (PI), is commonly implemented by decomposing it into a sequence of elementary two-qubit gates, such as CNOT gates.
Other important examples of PI two-qubit gates include iSWAP, $\sqrt{i\mathrm{SWAP}}$, and controlled-$Z$.
Similarly, the $n$-qubit controlled-$Z$ gate
\begin{align}\label{eq:CZ}
    \mathrm{CZ}_{n-1}
    \eq
    \mathbb{I} - 2\ket{1}\bra{1}^{\otimes n}
\end{align}
is PI, yet its standard realization is typically compiled into a sequence of lower-body gates requiring individual qubit control.
This gate is equivalent, up to single-qubit Hadamard gates, to the $n$-qubit Toffoli gate, a key primitive in quantum algorithms and quantum error correction.
More generally, important families of PI unitaries include gates such as $e^{i\theta Z^{\otimes n}}$, generated by $n$-body Hamiltonians such as $Z^{\otimes n}$.
Such unitary transformations are useful not only as primitives for quantum computation and quantum error correction, but also as representatives of a broader class of many-body operations that are difficult to realize by conventional compilation into few-body gates.
Many-body interactions provide a powerful resource for quantum computing, with applications ranging from Grover search \cite{Wang_MS_2001} and digital simulations of lattice gauge theories \cite{Ciavarella_etal_2021} to the simulation of fermionic systems \cite{Seeley_etal_2012} via Jordan-Wigner \cite{Jordan_Wigner_1928} or Bravyi-Kitaev \cite{Bravyi_Kitaev_2002} encodings, and to the simulation of many-body quantum dynamics \cite{Müller_etal_2011}.

\begin{figure*}
\centering{
\includegraphics[width=0.7\textwidth]{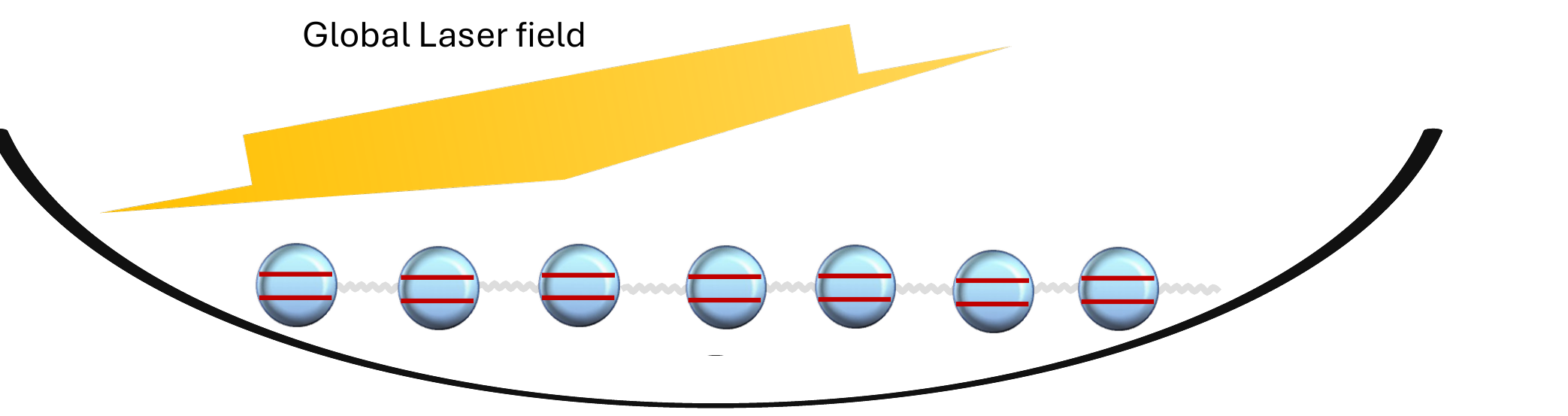}}
     \caption{\textbf{Example of global control in trapped ion systems.} By applying global laser fields, the collective motional modes of the ions can be coherently coupled to their internal degrees of freedom.
     Under suitable approximations, this interaction is effectively described by the TC Hamiltonian defined in \cref{eq:TCham}
     Importantly, this scheme eliminates the need for individual qubit addressing, significantly simplifying experimental implementation and enhancing scalability in trapped-ion quantum processors.}
 \label{fig:global_control}
\end{figure*}

\subsection*{Summary of main results}
The above discussion motivates a natural question: can two- and many-body PI gates be realized directly using only global controls, without addressing individual qubits?
In this work, we answer this question in the affirmative.
We consider the Tavis-Cummings (TC) model \cite{tavis_1968_exact_solut,TC2_1969}, the natural multi-qubit generalization of the Jaynes-Cummings model, in which several qubits couple to a common bosonic mode (oscillator).
Here we focus on the PI setting in which all qubits couple \textit{identically} to the bosonic mode.
The resulting TC Hamiltonian $\HTC$, defined in \cref{eq:TCham}, arises naturally, for example, when the internal degrees of freedom of ions in a linear chain couple to their center-of-mass vibrational mode, and has been studied extensively in atomic physics, quantum optics, and quantum control \cite{Bashir_Abdalla_1995,Bogoliubov_etal_1996,Rybin_etal_1998,Tessier_etal_2003,Vadeiko_etal_2003,Genes_etal_2003,Fink_etal_2009,Agarwal_etal_2012,Zou_etal_2013,Keyl_2014_control}.

Our first main result, stated in \cref{thm:mainThm}, shows that \emph{any PI unitary transformation on $n$ qubits can be realized using the TC interaction and global $z$ and $x$ fields, with the bosonic mode initialized in and returned to the vacuum state, or more generally in any eigenstate of $a^\dag a$, the occupation number operator (equivalently, the harmonic oscillator Hamiltonian)}.
As discussed above, using PI controls is practically appealing because it avoids individual qubit addressing: all control operations act identically on all qubits; see \cref{fig:global_control}.
Such global-control features have already been exploited experimentally to improve system performance; see, e.g., \cite{Lu_etal_2019,Katz_etal_2023_34qubit}.
As an immediate corollary, we prove that all PI states can be prepared using the TC  interaction and global fields; see \cref{cor:state_prep}.
This includes useful entangled states such as GHZ and Dicke states, which are important resources, for example, in quantum sensing and metrology.
We also show that, by combining this state-preparation result with a qubit-oscillator swap operation, one can prepare arbitrary states in the lowest $(n+1)$ energy levels of the bosonic mode; see \cref{cor:state_prep2}.
\vspace{0.4em}

A simple but important idea underlying our constructions is to use the global $x$ field only when the bosonic mode is decoupled from the qubits and has returned to its initial number state (see \cref{fig:x-z}).
This strategy is motivated by the U(1) symmetry of the TC interaction, associated with conservation of the total excitation number $a^\dag a+J_z$.
Consequently, when only $\HTC$ and $J_z$ are used, the joint Hilbert space decomposes into a direct sum of finite-dimensional invariant subspaces labeled by the total excitation number.
This decomposition is one of the basic simplifying features of the Jaynes-Cummings and Tavis-Cummings models, which has been used since its early days.

By contrast, the global $x$ field, corresponding to Hamiltonian $J_x$, does not conserve $a^\dag a+J_z$.
If applied while the bosonic mode is not in its initial state, it can couple the dynamics to arbitrarily high energy levels, making the problem genuinely infinite-dimensional.

This motivates a useful first step: before allowing global $x$ rotations, we characterize which qubit unitaries can be realized using only the TC interaction and the global $z$ field, with the bosonic mode serving as an ancillary system and returned to its initial number state.
There are also practical reasons to understand this restricted setting.
In platforms such as trapped ions and superconducting qubits, a global $x$ field does not commute with the intrinsic qubit Hamiltonian, which is oriented along the $z$ direction.
As a result, such non-energy-conserving operations can be more sensitive to certain noise sources, for instance fluctuations in the reference clock that sets the timing and phase of the pulses.
Therefore, it is interesting and useful to identify which operations can be realized without using an $x$ field.

These considerations lead to the following restricted-control problem: which PI qubit unitaries can be implemented using only $\HTC$ and $J_z$, with the bosonic mode initialized in and returned to a fixed number state?
Since these Hamiltonians conserve the total excitation number $a^\dag a+J_z$, any qubit unitary realized with the bosonic mode returned to its initial number state also conserves $J_z$. 
 
This raises a sharper question: \emph{can we realize all PI unitaries that also conserve $J_z$ using only $\HTC$ and $J_z$?}
Examples of such gates include $\mathrm{CZ}_{n-1}$ in \cref{eq:CZ} and unitaries generated by Hamiltonian $Z^{\otimes n}$.

\begin{figure*}
  \centering{
  \includegraphics[width=0.8\textwidth]{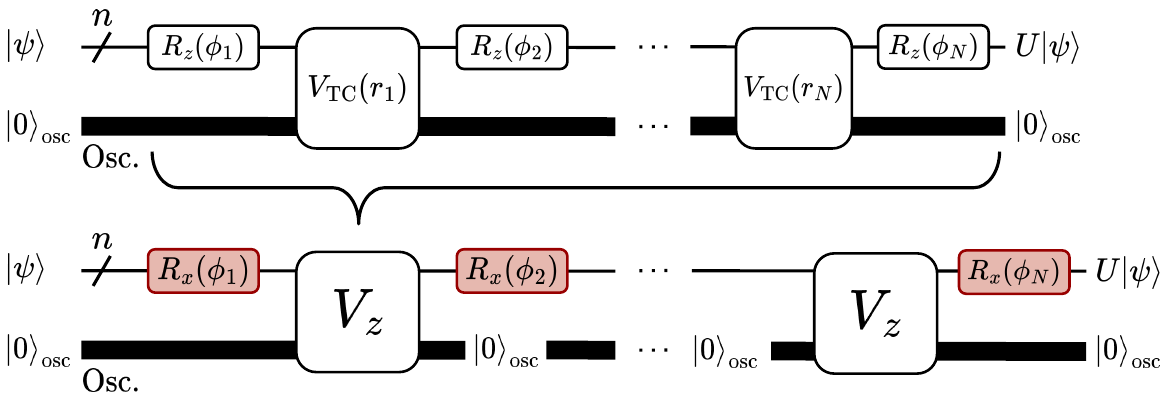}
  \caption{
    \textbf{Implementing multi-qubit PI      unitaries using the TC interaction.}
    In our constructions, the global field in the $x$ direction, i.e., the Hamiltonian $J_x$, is applied only after the bosonic mode has returned to its initial state, for example, the vacuum state $\vac$ shown in this figure.
    This prevents \textit{uncontrolled} leakage to arbitrarily high oscillator excitations.
    \textbf{Top:} A PI, $J_z$-conserving $n$-qubit unitary $U$ is realized using a global $z$ field and the TC interaction, which couples the qubits to a bosonic mode initialized in and returned to the vacuum state $\vac$, as in \cref{eq:ancilla}.
    \cref{thm:general_k_ancilla} characterizes the unitary transformations $U$ that can be realized on the qubits in this way.
    \textbf{Bottom:} Unitaries of the form shown in the top circuit are interleaved with global $x$ rotations, enabling the realization of arbitrary $n$-qubit PI unitaries; see \Cref{thm:mainThm}.}\label{fig:x-z}
}
\end{figure*}
Surprisingly, we find that the answer depends both on the number of qubits and on the initial number state of the oscillator (see \cref{thm:general_k_ancilla}).
In particular, if the bosonic mode is initially prepared in its ground state, i.e., the vacuum, which is the most common setting, then the answer is affirmative for $n=2$ and $n=3$ qubits, but not for $n\ge 4$.
For $n\ge4$, additional phase constraints appear; in particular, these constraints forbid the implementation of $\mathrm{CZ}_{n-1}$.
Interestingly, however, we find that the conjugated gate $X^{\otimes n}\mathrm{CZ}_{n-1}X^{\otimes n}$ does satisfy the constraint and is therefore realizable in this restricted setting.
This asymmetry originates from the fact that the vacuum is the lowest oscillator number state: states with minimal $J_z$ in each angular-momentum sector cannot exchange excitations with the oscillator and therefore can acquire only the phases generated by global $z$ rotations.

These unexpected constraints in the absence of $J_x$ also highlight why the main theorem (\cref{thm:mainThm}) is nontrivial.
Indeed, several independent obstructions might suggest that universal control over PI unitaries is impossible in this setting: the constraints identified in \cref{thm:general_k_ancilla}, the no-go theorem of \cite{marvian_sym_loc_2022}, which places severe restrictions on realizing PI unitaries using interactions that decompose into sums of $k$-body terms with $k<n$, and the accidental symmetry of the TC interaction identified in our companion paper \cite{symmetry_paper}.
The main theorem shows that, nevertheless, these obstructions can be overcome once global $x$ rotations are included.

\vspace{0.2em}
Finally, in addition to these general results, we provide explicit circuit constructions in the especially relevant case of $n=2$ qubits.
As mentioned above, in this case, all PI unitaries that conserve $J_z$ can be realized using $\HTC$ and $J_z$, with the oscillator initialized in and returned to the vacuum.
We provide exact circuits for several useful PI gates, including controlled-$Z$ (CZ), SWAP, iSWAP, and $\sqrt{i\text{SWAP}}$; see \Cref{tab:2Q_simple_times,tab:circuit_table}.
More generally, we present a synthesis method for implementing any two-qubit PI, $J_z$-conserving unitary using finite sequences of gates generated by $\HTC$ and $J_z$.
These circuits are exact, and their length, measured by the total time for which the TC interaction is turned on, is upper bounded by approximately $3.92\times 2\pi\gTC^{-1}$, where $\gTC$ is the TC coupling strength (see \cref{eq:TCham}).

\begin{table}[]
    \centering
    \renewcommand{\arraystretch}{1.5}
    \setlength{\tabcolsep}{6pt}
    \begin{tabular}{l|c}
        \textbf{Gate} & \shortstack{\textbf{Time}  $(\tau\times\gTC/2\pi)$} \\\hline
         CZ & $\approx {2.87}$\\
         SWAP & {$\approx1.27$}\\
         iSWAP & $\approx {2.55}$\\
         $\sqrt{\text{iSWAP}}$ & $\approx {2.68}$
    \end{tabular}
    \caption{\textbf{Interaction time for useful two-qubit PI gates.} We have derived explicit sequences of global $z$ fields and the TC interaction for the exact implementation, up to numerical precision, of the above useful two-qubit gates.
    Here, $\tau$ denotes the total time during which the TC interaction must be active, and $g_{\text{TC}}$ is the TC coupling strength; see \cref{sec:useful_2q_gate} for further details.}
    \label{tab:2Q_simple_times}
\end{table}

It is useful to compare this approach with two well-known methods for implementing controlled-$Z$, or equivalently CNOT up to single-qubit rotations, via qubit--oscillator interactions: the M{\o}lmer-S{\o}rensen \cite{Molmer_Sorensen} and Cirac-Zoller approaches \cite{Cirac_Zoller_1995}.
The M{\o}lmer-S{\o}rensen gate uses both JC and anti-JC interactions, corresponding to red and blue sidebands, and is robust against imperfect cooling of the oscillator to its ground state.
The Cirac-Zoller approach, by contrast, can use only the JC interaction, but requires individual control over the qubits.
Our two-qubit constructions use the TC interaction and global $z$ fields only: the qubits couple identically to the oscillator, and no individual addressing is required.

This paper is organized as follows.
In \cref{sec:setup}, we introduce the Tavis-Cummings Hamiltonian and formulate the problem of realizing qubit unitaries by coupling the qubits to an oscillator that returns to its initial number state.
In \cref{sec:main_results}, we present our main result, \cref{thm:mainThm}, and discuss applications to entanglement generation, PI state preparation, and oscillator state preparation.
In \cref{sec:TC_and_z}, we study the restricted setting with only  Hamiltonians $\HTC$ and $J_z$, characterize the resulting constraints, and explain why all PI, $J_z$-conserving gates are realizable for two and three qubits with the oscillator initialized in the vacuum.
Finally, in \cref{sec:useful_2q_gate}, we give explicit circuit constructions for two-qubit PI, $J_z$-conserving gates, including CZ, SWAP, iSWAP, and $\sqrt{i\mathrm{SWAP}}$.

\section{Setup: Tavis-Cummings Interaction}
\label{sec:setup}
We are interested in implementing permutation-invariant (PI) $n$-qubit unitary transformations, using global qubit control.
The simplest such operations are the uniform field Hamiltonians of the form $\omega_w(t)J_w$, where
\begin{align}
    J_w \,:=\, \frac{1}{2}\sum_{i=1}^{n} {\sigma}_w^{(i)}\quad:\quad w=x,y,z
 \label{eq:Zham}
\end{align}
are the $w$-components of total angular momentum on $n$ qubits, ${\sigma}_w^{(i)}$ is Pauli operator ${\sigma}_w$ on qubit $i$, and $\omega_w(t)$ are controllable qubit transition frequencies. 
These Hamiltonians generate rotations around the $x$, $y$, and $z$ axes, but they do not couple qubits together.

Furthermore, we assume that the qubits are coupled to a bosonic mode via the PI $n$-qubit version of the Jaynes-Cummings (JC) interaction.
This interaction is given by the Tavis-Cummings (TC) Hamiltonian,
\begin{align}
 \begin{split}
    \HTC &\eq \frac{\gTC}{2}\sum_{i=1}^{n} \Big(\sigma_+^{(i)}{a} + \sigma_-^{(i)}{a}^{\dag}\Big) \\[4pt]
    &\eq \gTC\,({J}_+{a} + {J}_-{a}^{\dag})\, ,
 \end{split}
 \label{eq:TCham}
\end{align}
where ${J}_{\pm}=J_x\pm i J_y$, $a$ is the annihilation operator on the bosonic mode Hilbert space satisfying $[a, a^\dag]=\mathbb{I}$, and $\gTC$ is a uniform coupling strength
\footnote{As usual, we often suppress tensor product in the notation, e.g. $J_+\otimes a\simeq J_+a$.}.
This Hamiltonian is commonly used, for example, as an effective approximation for describing the interaction between the motional modes of the ions and their internal degrees of freedom,  which serve as qubits \cite{Cirac_Zoller_1995,Molmer_Sorensen}.

We are interested in the family of unitaries $\V_{z-x}$ acting on the joint Hilbert space
\begin{align}
    \H_{\text{qubits}}\otimes\H_{\text{osc}} \,=\,(\C^2)^{\otimes n}\otimes\mathcal{L}^2(\R)\, ,
\end{align}
realized as \textit{quantum circuits} consisting of finite sequences of the three continuous one-parameter families of unitaries corresponding to each Hamiltonian:
\begin{align}
 \begin{split}
    R_z(\phi_z)&:=\exp(-i\phi_z J_z)\\[2pt]
    R_x(\phi_x)&:=\exp(-i\phi_x J_x)\\[2pt]
    \UTC(r)&:=\exp(-i r  \HTC /\gTC)\,,
 \end{split}
\end{align}
where $r\in\mathbb{R}$ and $\phi_z,\phi_x\in [-2\pi,2\pi)$  are dimensionless parameters.
(See, for example, \cref{tab:gadget_circuits}.)
\footnote{Since rotations around the $y$ axis can be synthesized from rotations around the $x$ and $z$ axes, we do not include them separately here.}
Note that
\begin{align}
    R_z(\pi)^\dag \HTC R_z(\pi)=-\HTC .
\end{align}
Thus, by applying a $\pi$ rotation about the $z$ axis, we can reverse the sign of $\HTC$, or equivalently the sign of the coupling $\gTC$.
It is therefore natural to allow $\UTC(r)$ with both positive and negative values of $r$.

Equivalently, $\V_{z-x}$ is the \textit{group} of unitary transformations that can be realized using the time- dependent Hamiltonian
\begin{equation}
    H_{z-x}(t) \,:=\, f(t) \HTC + \omega_z(t) J_z + \omega_x(t) J_x \,,
 \label{eq:totH}
\end{equation}
for arbitrary real functions $f(t)$ and qubit transition frequencies $\omega_z(t)$ and $\omega_x(t)$; under the Schr\"odinger equation $dV(t)/dt=-i H_{z-x}(t) V(t) : 0 \le t\le T $ for arbitrary time $T$, with the initial condition $V(0)=\mathbb{I}$.
\footnote{
The harmonic oscillator's intrinsic Hamiltonian $H_\text{osc}=\nu a^\dag a$, is often fixed and time-independent.
\Cref{eq:totH} represents the dynamics in the \textit{interaction picture} -- i.e. in a rotating frame given by transformation $\ket{\psi(t)}_{\text{Rot}} = \exp(i \nu (a^\dag a + J_z) t) \ket{\psi(t)}_{\text{Lab}}$, in which the oscillator has no intrinsic time evolution.
(See \cite{theory_paper} for further discussion.)}

\vspace{1mm}
The main goal of this paper is to characterize the unitary transformations on
$\H_{\text{qubits}}=(\mathbb{C}^2)^{\otimes n}$ that can be implemented by elements of $\mathcal{V}_{z-x}$ acting on $\H_{\text{qubits}}\otimes\H_{\text{osc}}$, while returning the bosonic mode to its initial state.
We assume the bosonic mode is initialized in its ground (vacuum) state $\ket{0}_{\text{osc}}$, or more generally, an eigenstate $\ket{k}_{\text{osc}}$ of the occupation number operator (or, equivalently, the intrinsic harmonic oscillator Hamiltonian) $a^{\dag}a$.
Then, to implement a desired $n$-qubit unitary transformation $U:(\mathbb{C}^2)^{\otimes n}\rightarrow (\mathbb{C}^2)^{\otimes n}$, one finds a unitary $V\in \mathcal{V}_{z-x}$ such that for any initial qubit state $\ket{\psi}\in(\mathbb{C}^2)^{\otimes n}$, it holds that
\begin{align}\label{eq:ancilla}
    V\big(\ket{\psi}\otimes \ket{k}_{\text{osc}}\big)=\big(U\ket{\psi}\big)\otimes \ket{k}_{\text{osc}}\,,
\end{align}
which can be equivalently expressed as $U=\ipo{k}{V}{k}_{\text{osc}}$.
We often consider implementing the unitary $U$ up to a global phase, which is physically irrelevant
\footnote{This is equivalent to adding the identity operator to the set of realizable Hamiltonians.}, meaning
\begin{align} \label{eq:ancilla2}
    e^{i\alpha}\ipo{k}{V}{k}_{\text{osc}} = U\,\,:\,\,\alpha\in[0,2\pi)\,,
\end{align}
for some phase $\al$.
Then, we construct quantum circuits implementing the corresponding unitary $V$ as sequences of unitaries $R_z(\theta)$, $R_x(\theta)$, and $V_{\text{TC}}(r)$.

\section{Permutation-Invariant (PI) unitaries from the TC interaction}
\label{sec:main_results}
Since all Hamiltonians $H_{z-x}(t)$ act identically on all qubits, all unitary transformations $V$ realized by such Hamiltonians are also PI. 
This, in turn, implies that any $n$-qubit unitary $U$ realized via \cref{eq:ancilla} is PI.
The question is, can we realize all PI unitaries in this way? Our first main result gives an affirmative answer to this question:

\begin{theorem}\label{thm:mainThm} All PI unitary transformations on $n$ qubits can be realized, up to a global phase, using the Hamiltonian $H_{z-x}(t)$ in \cref{eq:totH}, i.e., uniform global $z$ and $x$ fields $J_z$ and $J_x$ and the Tavis-Cummings interaction $H_{\text{TC}}$, where the bosonic mode (oscillator) is initialized in and returned to an arbitrary eigenstate $\ket{k}_{\text{osc}}$ of the intrinsic oscillator Hamiltonian $a^{\dag}a$.
\end{theorem}
\Cref{fig:x-z} gives an overview of the proof of this theorem, which is further expanded in  \cref{app:theorem_proof}.
In particular, a key ingredient is \cref{thm:general_k_ancilla} below, which characterizes the unitaries realizable using only the Hamiltonians $\HTC$ and $J_z$.
Combining these unitaries with global $x$ rotations, we obtain all PI unitaries.

\subsection{Applications: Entanglement Generation and Preparation of GHZ and Dicke states}
\label{sec:entanglement}
An important corollary of \cref{thm:mainThm} is that any two PI pure $n$-qubit states, equivalently any two states in the totally symmetric subspace of $(\C^2)^{\otimes n}$, can be transformed into one another using Hamiltonians $\HTC$, $J_x$, and $J_z$.
In particular, every such state can be prepared from the easily accessible initial state $\ket{0}^{\otimes n}$.
This includes useful states such as the GHZ state \cite{GHZ_1989}, $(|0\rangle^{\otimes n} + \ket{1}^{\otimes n})/\sqrt{2}$, and Dicke states, i.e., the eigenstates of $J_z$ in the totally symmetric subspace.
(Inside this subspace, $J_z$ is non-degenerate).
\begin{corollary} \textbf{(Entanglement generation)}
    Any desired state in the symmetric subspace of $n$ qubits can be obtained from the initial state $|0\rangle^{\otimes n}$, using the Tavis-Cummings interaction $H_{\text{TC}}$, which couples the qubits to an oscillator initialized in and returned to the vacuum state, along with uniform global $z$ and $x$ fields $J_z$ and $J_z$.
 \label{cor:state_prep}
\end{corollary}

\begin{proof}
This follows from \cref{thm:mainThm}, because any desired PI state $\ket{\psi} \in (\mathbb{C}^2)^{\otimes n}$ can be prepared by implementing a unitary that acts trivially on the entire Hilbert space except for the 2D subspace spanned by $\ket{\psi}$ and $ |0\rangle^{\otimes n}$.  
Any such unitary will be PI; explicitly, one can use for example, the direct rotation unitary $\sqrt{(I-2\ket{\psi}\langle\psi|) (I-2|0\rangle\langle 0|^{\otimes n})}$.
\end{proof}

We note that, prior to this work, the possibility of preparing states in the symmetric subspace was known only in a few specific cases.
For instance, \cite{Retzker_2007} demonstrates how to prepare Dicke states, albeit requiring the harmonic oscillator to be in excited states, whereas our result shows that such states can be prepared with the oscillator in the ground state.
See also \cite{Mu_etal_2020, Zhang_etal_2024} for methods of preparing Dicke and GHZ states, using the TC and anti-TC interactions, plus a Stark Hamiltonian $J_z\otimes\Nhat$.

\section{Respecting the U(1) symmetry: Realizable \\unitaries using Hamiltonians $\HTC$ and $J_z$}
\label{sec:U1_symmetry}
In this section, we characterize the unitary transformations that can be realized on the qubits using only the Hamiltonians $J_z$ and $\HTC$ (see
\cref{thm:general_k_ancilla}).
This characterization is needed for the proof of our main result, \cref{thm:mainThm}, but it is also useful and interesting in its own right: it clarifies what can be achieved without the additional control Hamiltonian $J_x$.
Unlike $J_z$ and $\HTC$, this control breaks the U(1)
symmetry and is, in some sense, more demanding, since it does not commute with the intrinsic qubit Hamiltonian, which we assume to be oriented along the
$z$ direction.

That is, we focus on the subgroup $\mathcal{V}_z\subset\mathcal{V}_{z-x}$ of unitaries generated by time evolution under Hamiltonians of the form
\begin{align}
    H_z(t) \,:=\, f(t)\HTC + \omega_z(t)J_z \,.
 \label{eq:totHz}
\end{align}
Equivalently, $\mathcal{V}_z$ is the set of quantum circuits realized as finite sequences of unitaries $R_z(\phi)$, with $\phi\in[-2\pi,2\pi)$, and
$V_{\text{TC}}(r)$, with $r\in\mathbb{R}$.

With the global $x$ field turned off, the dynamics generated by $H_z(t)$ conserve the operator $a^\dag a+J_z$, which can be interpreted as the total number of excitations in the qubit-bosonic system.
In other words, the time evolution respects the U(1) symmetry with unitary representation
\begin{align}
    \exp(i\theta J_z)\otimes \exp(i\theta a^\dag a)\quad:\quad \theta\in[0,4\pi)\,.
\end{align}
This symmetry implies that $\H_{\text{qubits}}\otimes\H_{\text{osc}}$
decomposes into a direct sum of finite-dimensional invariant subspaces, each labeled by a distinct eigenvalue of $a^\dag a+J_z$, or equivalently by the total number of excitations.

\vspace{0.3em}
In our companion paper \cite{theory_paper}, we fully characterize $\mathcal{V}_z$ when restricted to any finite, but arbitrarily large, collection of these invariant subspaces.
This characterization shows that the unitaries realizable using only
$\HTC$ and $J_z$ are subject to two distinct types of constraints beyond the obvious permutational and U(1) symmetries, i.e., the $\Uo\times\Sn$ symmetry.
The first are generic central constraints: the center of $\mathcal{V}_z$ is smaller than the center of the full group of unitaries respecting these
symmetries, so not all relative phases between inequivalent symmetry sectors can be independently controlled.
The second type arises from an additional ``accidental'' symmetry of the TC interaction, which forces the realized unitaries in certain invariant sectors to coincide for $n\geq 3$.
In the companion paper \cite{symmetry_paper}, we explain this unexpected symmetry using Schwinger's oscillator model of angular momentum and analyze its consequences for controllability.
We also show that these accidental constraints are not fundamental consequences of the $\Uo\times\Sn$ symmetry: they can be removed by adding the PI and U(1)-invariant Hamiltonian $J_z^2$.
In particular, $\HTC$, $J_z$, and $J_z^2$ generate all PI, U(1)-invariant unitary transformations, up to the generic central constraints.

\subsection{Which qubit unitaries are realizable using Hamiltonians $\HTC$ and $J_z$?}
\label{sec:TC_and_z}
We now return to the main goal of this section: characterizing the $n$-qubit unitaries $U$ on $\H_{\text{qubits}}=(\mathbb{C}^2)^{\otimes n}$ that can be realized using Hamiltonians $\HTC$ and $J_z$.
More precisely, we ask which such unitaries can be implemented by applying a unitary $V\in\mathcal{V}_z$ to the joint qubit-bosonic system, in such a way that
\begin{align} \label{eq:ancilla_2}
    V\big(\ket{\psi}\otimes \ket{k}_{\text{osc}}\big)=\big(U\ket{\psi}\big)\otimes \ket{k}_{\text{osc}}\,,
\end{align}
for a fixed eigenstate $\ket{k}_{\text{osc}}$ of the bosonic mode, but arbitrary unknown state $\ket{\psi}\in \H_{\text{qubits}}$.

Since the bosonic mode returns to its initial state at the end of the evolution, $a^\dag a$ is conserved, which in turn implies that $J_z$ must also be conserved, that is
\begin{align}
    [U, J_z]=0\,.
\end{align}
In other words, such unitaries respect the U(1) symmetry corresponding to rotations around $z$, i.e.,  $\exp(i\theta J_z) : \theta\in[0,4\pi)$.
Such unitary transformations are also sometimes called energy-conserving unitaries, as they preserve the total intrinsic Hamiltonian of the qubits, which is often proportional to $J_z$, the sum of Pauli $Z$ operators.

The question here is: can we realize \textit{any} $n$-qubit PI unitary $U$ that also conserves $J_z$, such as the $\text{CZ}_{n-1}$ gate, in this way? 
In other words, can we avoid using Hamiltonian $J_x$, which breaks this symmetry, when implementing them?
As summarized in \cref{thm:general_k_ancilla} below, the characterization of $\V_z$ in \cite{theory_paper} shows that the answer depends on the initial state of the bosonic mode.
If the bosonic mode is initialized in an eigenstate $\ket{k}_{\text{osc}}$ of $a^\dag a$ with $k\geq 1$, then any such unitary can be realized for all $n\geq 2$.
By contrast, if the bosonic mode is initialized in the vacuum state $\ket{0}_{\text{osc}}$, then this remains true for $n=2$ and $n=3$, but fails for $n\geq 4$.

Before presenting the formal result in \cref{thm:general_k_ancilla}, it is useful to first characterize the group of all unitaries on $(\mathbb{C}^2)^{\otimes n}$ that respect both permutational symmetry and the U(1) symmetry generated by $J_z$, i.e., those that commute with both the permutation action and $J_z$.
This group consists precisely of unitaries of the form
\begin{align}
    U \eq \sum_{j=j_{\min}}^{n/2}\sum_{m=-j}^je^{i\phi_{j,m}} P_{j,m}\quad:\quad\phi_{j,m}\in[0,2\pi)\,,
 \label{eq:charU}
\end{align}
where $P_{j,m}$ is the projector onto the subspace of $(\mathbb{C}^2)^{\otimes n}$ that is a simultaneous eigenspace of $J_z$, with eigenvalue $m$, and of the total angular momentum operator $J^2=J_x^2+J_y^2+J_z^2$, with eigenvalue $j(j+1)$.
Here, $j$ takes integer values $0,\dots, n/2$, when $n$ is even, and half-integer values $1/2,\dots ,n/2$  when $n$ is odd.
Furthermore, for general PI, U(1)-invariant unitaries, the phases $\phi_{j,m}$ are arbitrary and independent of each other.
As we further explain in \cref{app:schur_weyl_basis}, this characterization follows immediately from Schur-Weyl duality, together with the fact that $J_z$ is non-degenerate in each irrep of SU(2).

The following theorem, which is adapted from Theorem 9 and Corollary 10 of our companion paper \cite{theory_paper}, characterizes the subgroup of such unitaries realizable using only Hamiltonians $\HTC$ and $J_z$.

\begin{theorem} \label{thm:general_k_ancilla}
Consider a unitary transformation $U$ acting on $\H_{\text{qubits}}=(\mathbb{C}^2)^{\otimes n}$, and let $\ket{k}_{\mathrm{osc}}$ be an eigenstate of the number operator $a^\dag a$ with integer eigenvalue $k\geq 0$.
Then there exists a unitary $V\in\V_z$, i.e., a unitary realizable using Hamiltonians $\HTC$ and $J_z$, and global phase $e^{i\alpha}$, such that
\begin{align}
    V\big(\ket{\psi}\otimes\ket{k}_{\mathrm{osc}}\big) \eq 
    e^{-i\alpha}\big(U\ket{\psi}\big)\otimes\ket{k}_{\mathrm{osc}}
\end{align}
for all $\ket{\psi}\in(\mathbb{C}^2)^{\otimes n}$, if and only if the following conditions hold:
\begin{enumerate}
\item[(1)] $U$ is PI and preserves $J_z$, i.e., it is of the form given in \cref{eq:charU}.
\item[(2)] If $k=0$, i.e., the bosonic mode is in vacuum, then there exists $\beta\in[-2\pi,\,2\pi)$ such that
\begin{align} \label{eq:const0_main}
    \phi_{j,-j} = \al+j\beta \,(\mathrm{mod}\,2\pi)\quad\text{for all}\,\,j\,.
\end{align}
\end{enumerate}
\end{theorem}
In summary, we find that if the number of initial excitations $k\geq1$ , then all PI, U(1)-invariant unitaries on $n$ qubits can be realized with $\HTC$ and $J_z$.
However, as we discuss next, this is no longer true for $k=0$, i.e., when the bosonic mode is initialized in its vacuum state.

\subsection{The vacuum is restrictive: \\Unexpected constraints for $n\ge 4$ qubits}
According to this theorem, when the bosonic mode is initialized in the vacuum, the phases $\phi_{j,-j}$ associated with $m=-j$ are not independent: for all values of $j$, they are determined solely by a global phase $\alpha$ and a global rotation around $z$ by angle $\beta$, and hence by only two free parameters.
Thus this condition imposes additional restrictions on the realizable qubit unitaries precisely when the qubit Hilbert space contains more than two distinct angular-momentum sectors $j$.
For $n$ qubits, $j$ takes $\lfloor n/2 \rfloor+1$ distinct values.
Hence there are $\lfloor n/2 \rfloor+1$ phases $\phi_{j,-j}$, and additional restrictions arise precisely when
\[\lfloor n/2 \rfloor + 1 \,>\, 2\,,\]
or, equivalently, when $n\ge 4$.
For example, these constraints forbid the implementation of the family of CZ gates
\begin{align*}
    \mathrm{CZ}_{n-1} \eq \mathbb{I}-2\ket{1}\bra{1}^{\otimes n}
\end{align*}
for $n\ge 4$.
Indeed, this gate has a single nonzero phase: $\phi_{j,m}=\pi$ in the sector $j=n/2$, $m=-n/2$, while all other phases vanish.
Such a phase pattern cannot be generated using only the two free parameters associated with the global phase and global rotations around $z$ once more than two distinct angular-momentum sectors are present.

On the other hand, interestingly, the anti-controlled-Z gate,  
\begin{align}
    X^{\otimes n}\text{CZ}_{n-1} X^{\otimes n} \eq \mathbb{I} - 2|0\rangle\langle 0|^{\otimes n} \,,
\end{align}
satisfies the constraint in \cref{eq:const0_main}, and it is therefore realizable without breaking the U(1) symmetry.
(For this unitary, $\phi_{j,m}=0$ for all $j$ and $m$ except $m=j=n/2$, which corresponds to state $\ket{0}^{\otimes n}$.)

The asymmetry between the states $\ket{0}^{\otimes n}$ and $\ket{1}^{\otimes n}$, which forbids $\text{CZ}_{n-1}$ but allows $X^{\otimes n}\text{CZ}_{n-1} X^{\otimes n}$ is related to the fact that under time evolution of $\HTC$, the sum of $J_z$ and the number of excitations $a^\dag a$ in the bosonic mode is conserved. 
Additionally, the vacuum state $\ket{0}_{\text{osc}}$ has the minimum number of excitations in the bosonic mode. 
Consequently, when $J_z$ has its minimum value, it cannot change under $V_{\text{TC}}(r)$, meaning that for any $j$, state $\{\ket{j,m=-j}\otimes\vac\}$ is an eigenstate of $\HTC$ with a eigenvalue zero.
Therefore, it only acquires phases allowed under global rotations generated by Hamiltonian $J_z$.
More generally, the constraints in \cref{eq:const0_main} are a consequence of this restriction together with the permutational symmetry of the Hamiltonian, which implies that under $H_z(t)$, angular momentum $j$ remains conserved, i.e., $H_z(t)$ commutes with $J^2$.
\footnote{Note that the conservation of the angular momentum operator $J^2$  is not a consequence of SU(2) symmetry; indeed, this symmetry is broken in our system. Rather, it is a consequence of permutational symmetry.} 
Therefore, in each angular momentum $j$ subspace, the state with the lowest allowed value of $J_z$ will be affected by the above restriction.  

As one might expect from this discussion -- and as stated in \cref{thm:general_k_ancilla} -- the restrictions in \cref{eq:const0_main} disappear if instead of the vacuum, the oscillator is initialized in any other eigenstate of $a^\dag a$ with at least one excitation.
See \cref{app:theorem_proof} for details.

\subsection{The special case of two and three qubits}
In the special cases of $n=2$ and $n=3$ qubits, $j$ takes only two values: $j=0, 1$ for $n=2$, and $j=\sfrac{1}{2},\,\sfrac{3}{2}$ for $n=3$.
It follows that in both cases, \cref{eq:const0_main} does not impose any restrictions on the allowed phases $\phi_{j,-j}$; that is, they can take arbitrary values. 
Since, according to \cref{thm:general_k_ancilla}, these are the only restrictions on realizable PI, U(1)-invariant unitaries, we conclude that for $n=2$ and $n=3$, \textit{all} such unitaries are realizable in this way.

Indeed, in the case of $n=2$, we find explicit circuits and synthesis techniques for realizing arbitrary two-qubit gates that respect both U(1) and permutational symmetries.
From these circuits, we can determine the maximum required interaction time $\tau$ for which the qubits must interact with the oscillator via $\HTC$ -- that is, the total time the TC interaction must be active.
This can be quantified by $\tau := \int_0^T |f(t)| dt$, where $f(t)$ is the coefficient of $H_{\text{TC}}$ in Hamiltonian $H_z(t)$ in \cref{eq:totHz}, and $T$ is the total duration for implementing the desired unitary. 
In particular, $\tau$ will be the actual total time that the TC interaction is turned on, assuming $f(t)$ takes one of the three values $0$, $1,$ or $-1$.
\begin{theorem}[Two-qubit and three-qubit gates]
\label{thm:23_qubits}
Any PI, $J_z$-conserving unitary acting on $n=2$ or $n=3$ qubits can be realized exactly, up to a global phase, using $\HTC$ and $J_z$, with the bosonic mode initialized in and returned to an arbitrary eigenstate of $a^\dag a$.
Furthermore, when the bosonic mode is initialized in the vacuum, any such two-qubit unitary can be realized with total interaction time bounded by
\begin{align} \label{eq:time-bound}
    \tau := \int_0^T |f(t)|\,dt \,\,\leq\,\, 3.92 \times 2\pi \gTC^{-1}\,.
\end{align}
\end{theorem}
As we explained above, the first part of this theorem follows immediately from \cref{thm:general_k_ancilla}.
In \cref{sec:useful_2q_gate} and \cref{app:A_gate}, we demonstrate explicit circuit synthesis methods for implementing any such desired unitary, from which the time bound in \cref{eq:time-bound} follows.  
It is worth emphasizing that many useful gates -- such as controlled-Z (CZ), SWAP, and iSWAP -- can be realized with significantly shorter interaction times, as illustrated in \cref{tab:2Q_simple_times} and \cref{sec:useful_2q_gate}.
Furthermore, note that both CZ and iSWAP gates can transform a product state to a maximally entangled state. 
Therefore, using a CZ gate, for instance, one can generate maximally entangled states with total interaction time $\tau= 2.87 \times 2\pi \gTC^{-1}$.

\begin{table*}[htp]
 \centering
 \begin{adjustbox}{width=0.95\textwidth}
    \includegraphics[]{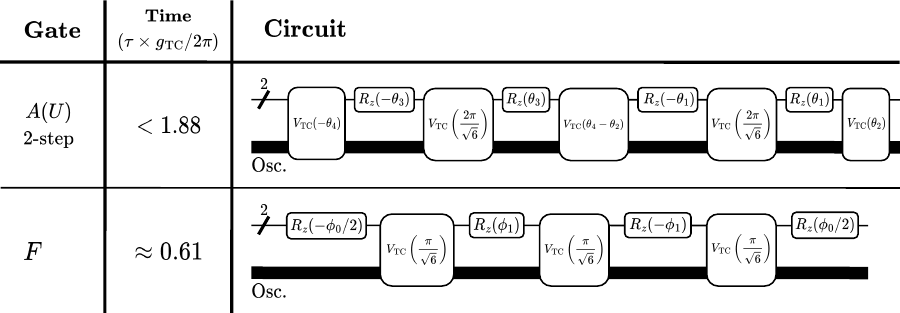}
 \end{adjustbox}
    \caption{\textbf{Circuits and (maximal) interaction times for $A(U)$ and $F$-gates:} Qubits are represented by the thin wire, while the bosonic mode (oscillator) is represented by the thicker wire.
    Fixed parameters for the $F$-gate are $\phi_1=\frac{1}{2}\arccos\left(\frac{7}{16}\right)$ and $\phi_0=\arctan\big(-\sqrt{23}/3\big)+\pi$.
    This $A(U)$ circuit represents a 2-step $A$-gate, from which any $A$-gate can be constructed, and angles $\theta_1,\theta_2,\theta_3,\theta_4$ depend on $U$. (See \cref{app:A_gate} for details.)}
    \label{tab:gadget_circuits}
\end{table*}
\subsection{Diagonal  PI $n$-qubit unitaries}
\label{sec:diag_unis}
As an important example of PI, $J_z$-conserving unitaries, we consider the important example of PI unitaries that are diagonal in the computational basis $\{\ket{0},\ket{1}\}^{\otimes n}$, which include the multi-controlled $Z$ gate, $\text{CZ}_{n-1}$.
Such unitaries which take the form
\begin{align}
    U &\eq \sum_{b_1,\dotsc, b_n=0}^1 \exp(i\phi(b_1,\dotsc, b_n))\,|b_1,\dotsc, b_n\rangle\langle b_1,\dotsc, b_n|\,,
\end{align}
where $b_j\in\{0,1\}$.
Permutational invariance means that $\phi(b_1,\dotsc, b_n)$ is uniquely determined by the sum $b_1+\dotsc+b_n$,  or equivalently, by the eigenvalue of $J_z$.
It follows that $U$ is a diagonal PI unitary if, and only if, it can be written as
\begin{align} \label{eq:diag_decomp}
    U \eq \sum_{m=-n/2}^{n/2} e^{i\phi_m} P_m\quad:\quad\phi_m\in[0,2\pi)\,, 
\end{align}
where
\begin{align}
    P_m &:= \sum_{j=|m|}^{n/2} P_{j,m} =\hspace{-3mm} \sum_{\substack{b_1,\dotsc,b_n\in\{0,1\}\\
    b_1+\cdots+b_n = n/2-m}}\hspace{-3mm}\ket{b_1,\dotsc,b_n}\bra{b_1,\dotsc,b_n}\,,
\end{align}
is the projector to the eigensubspace of $n$-qubit operator $J_z$ with eigenvalue $m\in[-n/2,\,n/2]$.
Here, we used the fact that in the subspace of $n$ qubits with total angular momentum $j$, i.e., the eigensubspace of $J^2$ with eigenvalue $j(j+1)$, $J_z$ takes eigenvalues $-j,\dotsc, j$, which means $j=|m|$ is the lowest value of $j$ such that $J_z$ has eigenvalue $m$ in that subspace.

\vspace{0.5em}
Applying the constraints in \cref{thm:general_k_ancilla} to the decomposition in \cref{eq:diag_decomp}, we prove in \cref{app:comp_basis_diag} the following corollary:
\begin{corollary}
\label{cor:diag}
A PI unitary $U$ that is diagonal in the computational basis is realizable using Hamiltonians $\HTC$ and $J_z$, with the bosonic mode initialized in and returned to the vacuum state $\vac$, if and only if there exist $\alpha\in[-\pi,\pi)$ and $\beta\in[-2\pi,2\pi)$ such that
\begin{align}
    \phi_m &=
    \begin{cases}
        \alpha+m\beta \,\,(\mathrm{mod}\,2\pi), & \text{for } m\leq 0, \\[6pt]
        \text{arbitrary}, & \text{for } m>0\,.
    \end{cases}
\label{eq:phi_m}
\end{align}
\end{corollary}
Again, we see an interesting asymmetry between negative and positive values of
$m$.
For $m>0$, the phases $e^{i\phi_m}$ are unrestricted, whereas for $m\leq 0$, they are determined by only two free parameters, $\alpha$ and $\beta$.
This asymmetry follows from the fact that, when $m\leq 0$, the phase $e^{i\phi_m}$ appears in the sector with $j=|m|$ and $m=-j$.
Hence it is subject to condition 2 of \cref{thm:general_k_ancilla}.

\section{Exact Pulse Sequences for 2-qubit PI gates}
\label{sec:useful_2q_gate}
In this section, we focus on the important case of $n=2$ qubits and introduce explicit pulse sequences for implementing all PI, U(1)-invariant (also known as energy-conserving) two-qubit gates using the TC interaction and global $z$ field.
For the remainder of this section, we assume that the bosonic mode is initialized in and returned to the vacuum state $\vac$, as in \cref{eq:ancilla} with $k=0$.
As examples, we present novel implementations of controlled-Z (CZ), SWAP, iSWAP, and $\sqrt{\text{iSWAP}}$ gates, and time evolution of the ZZ interaction.

We note that \cite{GomezRosas_etal_2021} demonstrates the approximate realization of two-qubit entangling operations via the TC interaction, albeit assuming that the oscillator is in a coherent state with a large expected number of quanta (e.g. 100 phonons).

\subsection{Two useful gadgets}
To achieve this, we first introduce two useful circuit modules that can be used to construct more complex circuits.
The first module is what we call an $A$-type gate: for any $2\times 2$ unitary $U\in\SU(2)$, $A(U)$ is an alternating sequence of at most seven $\UTC$ and six $\UZ$ gates that together realize the unitary $U$ in the 2D subspace spanned by
\begin{align}\label{eq:2D}
    \left\{\ket{00}\otimes \ket{1}_{\text{osc}}\,\,,\,\,\ket{\Psi^+}\otimes \vac\right\}\,,
\end{align}
whereas it acts trivially in the 3D subspace spanned by
\begin{align} \label{eq:3D}
    \left\{\ket{00}\otimes \vac\,\,,\,\, \ket{\Psi^+}\otimes \ket{1}_{\text{osc}} \,\,,\,\,  \ket{11}\otimes \ket{2}_{\text{osc}}\right\}\,,
\end{align}
where
\begin{align}
    \ket{\Psi^{\pm}}=\frac{|01\rangle\pm|10\rangle}{\sqrt{2}}\,.
\end{align}
Note that the subspaces in \cref{eq:2D} and \cref{eq:3D} are, respectively, the eigen-subspaces with eigenvalues $q=1$ and $q=2$ of the \textit{charge operator} $Q=a^\dag a+J_z+1$, which corresponds to the total number of \emph{excitations} in the system.
Furthermore, both subspaces live in the symmetric subspace of two qubits, i.e. the subspace corresponding to $j=1$, also known as the triplet subspace. 
Similar to all other unitaries that can be realized with $V_{\text{TC}}$ and $R_z$ gates alone, $A(U)$ gates act trivially in the anti-symmetric subspace of two qubits, corresponding to $j=0$, i.e. the singlet state $\ket{\Psi^-}$.

In particular, in \cref{app:A_gate} we present a construction that realizes $A(U)$ with an interaction time that depends on the Bloch sphere rotation angle associated with the unitary $U$.
Specifically, for $U = \exp(-i {\theta}\vec{\sigma} \cdot \hat{n} / 2)$, where $\hat{n} \in \mathbb{R}^3$ is an arbitrary unit vector, if $|\theta| \le 0.62\pi$, then $A(U)$ can be realized using the circuit given in \cref{tab:gadget_circuits}, for appropriately chosen parameter values, and with the total interaction time bounded by approximately $1.88 \times 2\pi\gTC^{-1}$. 
For arbitrary rotation angles $|\theta|$, the maximal required interaction time is $\approx2.69\times2\pi\gTC^{-1}$.
(See \cref{app:A_gate} for further details.) 

\begin{table*}[htp]
 \centering
 \begin{adjustbox}{width=0.75\textwidth}
    \includegraphics[]{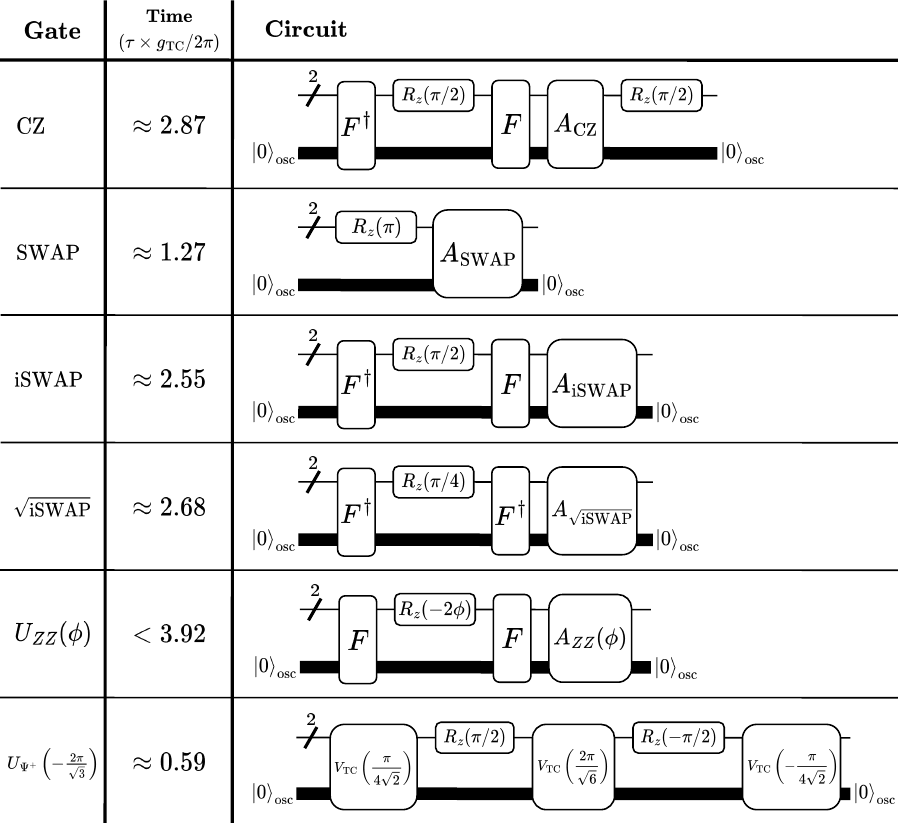}
 \end{adjustbox}
    \caption{Circuits and interaction times of six useful 2-qubit gates.
    For SWAP, the required $A$-gate is precisely a circuit of the form depicted in \cref{tab:gadget_circuits}, whereas the other $A$-gates are variations on this circuit.
    (See \cref{app:A_gate} for details.)
    The $A$-gate used $U_{ZZ}(\phi)$ has a variable interaction time that depends on $\phi$. Note that these circuits implement SWAP, iSWAP, $\sqrt{\text{iSWAP}}$, and $U_{ZZ}(\phi)$ up to global phases ($-1$, $i$, $e^{i\pi/4}$, and $e^{-i\phi}$, respectively).}
    \label{tab:circuit_table}
\end{table*}

The second useful circuit module, which will be called an $F$-gate, is realized via the circuit in \cref{tab:gadget_circuits}.
In contrast to $A$-gates introduced above, an $F$-gate acts non-trivially in the 3D subspace in \cref{eq:3D} by transferring two \emph{excitations} between the qubits and the bosonic mode:
\begin{align*}
    \ket{00}\otimes \vac\,  \xleftrightarrow{\quad F\quad }\,\ket{11}\otimes \ket{2}_{\text{osc}}\,.
\end{align*}
More precisely, under the action of $F$,
\begin{align}
 \label{eq:F_map}
 \begin{split}
    F\big(\ket{00}\otimes\vac\big)&=\ket{11}\otimes \ket{2}_{\text{osc}} \,,\\
    F \big(\ket{11}\otimes \ket{2}_{\text{osc}}\big)&= \ket{00}\otimes \vac\,,\\
    F\big(|\Psi^-\rangle\otimes\vac\big)&=|\Psi^-\rangle\otimes\vac \,, \\
    F\big(\ket{11}\otimes\vac\big)&=\ket{11}\otimes\vac\,.
 \end{split}
\end{align}
Here, the last two equations follow from the fact that $|\Psi^-\rangle \otimes \vac$ and $\ket{11} \otimes \vac$ are eigenstates of $\HTC$ with eigenvalue zero.
Roughly speaking, the first equation indicates that by applying $F$, we can transfer two excitations from the qubits to the bosonic mode, effectively shielding them from subsequent unitary transformations on the qubits.
We use this intuition in the next section to construct more general unitary transformations (see \cref{eq:seq2}).

One can see from the circuit in \cref{tab:gadget_circuits} that the total interaction time needed to implement an $F$-gate, i.e., the total time that the TC interaction must be turned on, is
\begin{align}
    \tau_F = \frac{3\pi}{\sqrt{6}}\times\gTC^{-1} \,\approx\, 0.612 \times 2\pi \gTC^{-1}\,.
\end{align}
It is worth noting that the $F$-gate also acts non-trivially in the 2D sector in \cref{eq:2D}.
This, in particular, means that the initial state $\ket{\Psi^+}\otimes \vac$ will be transformed to a state in which qubits are entangled with the bosonic mode.
This undesirable effect can be canceled by applying an appropriate $A$-gate. (See \cref{app:F_gate} for further details.)

\subsection{Implementing all 2-qubit PI, $J_z$-conserving unitaries}
Interestingly, all 2-qubit PI, $J_z$-conserving unitaries, such as the CZ, SWAP, iSWAP, and $\sqrt{\text{iSWAP}}$ gates, can be realized using two $F$ and/or $F^{\dag}$ gates, an appropriate $A$-type gate, and two $R_z$ gates.
To show this, first consider the sequence
\begin{align} \label{eq:F_combo}
    FR_z(\theta)F\,.
\end{align}
Thanks to the convenient form of the action of the $F$ gate in \cref{eq:F_map}, it is straightforward to see that this sequence implements the transformation
\begin{align}\label{eq:seq2}
 \begin{split}
  \renewcommand{\arraystretch}{1.2}
  \begin{array}{rlr}
    \ket{11}\otimes\vac &\longrightarrow &e^{i\theta}\ket{11}\otimes\vac\\
    \ket{00}\otimes\vac &\longrightarrow &e^{i\theta}\ket{00}\otimes\vac
    \\\ket{\Psi^-}\otimes\vac &\longrightarrow &\ket{\Psi^-}\otimes\vac\,.
  \end{array}
 \end{split}
\end{align}
In other words, it applies the same phase $e^{i\theta}$ to $\ket{00}\otimes\vac$ and $\ket{11}\otimes\vac$, which would otherwise obtain conjugate phases under $R_z(\theta)$.
\footnote{Note that $F$ and $F^{\dag}$ both implement the transformations in \cref{eq:F_map}, and therefore they are interchangeable in \cref{eq:F_combo}, meaning any combinations of $F$ and $F^{\dagger}$ in \cref{eq:F_combo} produces \cref{eq:seq2}.
However, different combinations transform state $\ket{\Psi^+}\otimes\vac$ differently, which means the $A$-gate required to remove this effect will be different. 
Therefore, this choice is relevant for determining the most optimal $A$-gate.
(See \cref{app:minimizing_2qubit_unis} for details.)}

On the other hand, under the action of $F R_z(\theta) F$ on the initial state $\ket{\Psi^+} \otimes \vac$,  the qubits become entangled with the bosonic mode.
This entanglement can be removed using a proper $A$-type gate.
In particular, for arbitrary phase $\theta_+\in[0,2\pi)$, we construct gate $ A(\theta, \theta_+)$ such that $ A(\theta, \theta_+)F R_z(\theta) F$ transforms this initial state to
\begin{align}
    \ket{\Psi^+}\otimes\vac\longrightarrow e^{i\theta_+}\ket{\Psi^+}\otimes\vac\,\,.
\end{align}
Note that since the $A$-type gate acts trivially on all states in \cref{eq:seq2}, the effect of $A(\theta, \theta_+) F R_z(\theta) F$ on these states is also described by \cref{eq:seq2}.
Therefore, for instance by choosing $\theta_+=0$, we realize the unitary $U_{ZZ}(\theta)=\exp(-i{\theta} Z_1 Z_2)$, the time evolution of a ZZ interaction.

In \cref{eq:seq2}, the phases $e^{i\theta}$ obtained by states $\ket{11}\otimes\vac$ and $\ket{00}\otimes\vac$ are identical. 
To make these phases independent we can apply $R_z(\theta')$, which gives these states opposite phases, namely $e^{i\theta'}$ and $e^{-i\theta'}$, respectively.
It follows that by applying the sequence 
\begin{align}
 \label{eq:2qubit_decomp}
    R_z(\theta') A(\theta, \theta_+) F R_z(\theta) F\,,
\end{align}
we can realize any arbitrary 2-qubit PI, U(1)-invariant unitary, up to a global phase.
Recall that such unitaries are characterized by \cref{eq:charU}, and up to a global phase, they are specified by 3 phases in the 2-qubit case.
In particular, \cref{eq:2qubit_decomp} implements the transformation:
\begin{align}\label{eq:2q_gate_map}
 \begin{split}
  \renewcommand{\arraystretch}{1.2}
  \begin{array}{rlr}
    \ket{11}\otimes\vac &\longrightarrow &e^{i(\theta+\theta')}\ket{11}\otimes\vac\\
    \ket{\Psi^+}\otimes\vac &\longrightarrow &e^{i\theta_+}\ket{\Psi^+}\otimes\vac\\
    \ket{00}\otimes\vac &\longrightarrow &e^{i(\theta-\theta')}\ket{00}\otimes\vac
    \\\ket{\Psi^-}\otimes\vac &\longrightarrow &\ket{\Psi^-}\otimes\vac\,.
  \end{array}
 \end{split}
\end{align}
We conclude that an arbitrary PI, U(1)-invariant two-qubit unitary $U$ can be realized, up to a global phase, using two $F$ and/or $F^{\dag}$ gates and one $A$-gate, with the maximum interaction time
\begin{align}
 \begin{split}
    \tau_{A} + 2\tau_{F}&\,\lesssim \,  3.92\times2\pi\gTC^{-1}\,.
 \end{split}
 \label{eq:2qubit_total_time2}
\end{align}
\Cref{tab:circuit_table} illustrates explicit circuits for several gates, including CZ, SWAP, iSWAP, and
\begin{align}
 \sqrt{\text{iSWAP}}&:= \exp\left(\frac{i\pi}{8}[X \otimes X + Y \otimes Y]\right)\,,
\end{align}
and gives the corresponding total interaction times.

\subsection{Fastest entangling gate?}
An interesting question is determining the minimum time required to implement an entangling unitary transformation on the qubits such that, at the end of the process, the bosonic mode is no longer entangled with the qubits. 
Here, we provide a candidate for such gates.
In particular, the unitary $U_{\Psi^+}(-2\pi/\sqrt{3})=\exp\big(i\pure{\Psi^+}2\pi/\sqrt{3}\big)$ is realized by the corresponding circuit in \cref{tab:circuit_table} with total interaction time 
\begin{align}
    \tau \approx 0.59 \times 2\pi{\gTC}^{-1}\,,
\end{align}
which is much shorter than the times required by the general methods discussed above.

\section{State Preparation Protocol for a Bosonic Mode}
\label{sec:oscillator_state_prep}
While the main focus of this paper is the implementation of PI gates on qubits, with the bosonic mode used as an ancillary system, the same tools also lead to a natural state-preparation protocol for the bosonic mode.
This protocol combines two ingredients: first, the ability to prepare arbitrary states in the totally symmetric subspace of the qubits, as shown in \cref{sec:entanglement}; and second, a state-transfer gadget introduced in our companion paper \cite{theory_paper}, which swaps a state encoded in the symmetric qubit subspace with the corresponding state of the bosonic mode.

More precisely, this gadget, which can be realized using only Hamiltonians $\HTC$ and $J_z$, implements the transformation
\begin{align}
    \Vswap\big(\ket{\psi}\otimes\vac\big)
    \eq
    \ket{1}^{\otimes n}\otimes\ket{\Psi}_{\rm osc}\, .
\end{align}
Here $\ket{\psi}\in(\mathbb{C}^2)^{\otimes n}$ is an arbitrary qubit state in the totally symmetric subspace, and $\ket{\Psi}_{\rm osc}$ is the corresponding state of the bosonic mode.
This correspondence is defined by expanding $\ket{\psi}$ in the Dicke basis, as follows.

Consider the Dicke basis, i.e., the eigenvectors of $J_z$ in the totally symmetric subspace,
\begin{align} \label{eq:Dicke} 
    \big|D_{n/2-m}^{(n)}\big\rangle
    &:=
    \binom{n}{\frac{n}{2}-m}^{-1/2}
    \hspace{-12pt}\sum_{\substack{b_1,\dotsc,b_n=0,1 \\[2pt]
    b_1+\dotsb+b_n = \frac{n}{2}-m}}
    \ket{b_1,\dotsc,b_n}\,,
\end{align}
for $m=-n/2,\ldots,n/2$.
In other words, $\big|D_{n/2-m}^{(n)}\big\rangle$ is the normalized symmetric superposition of all computational basis states with $\frac{n}{2}-m$ qubits in state $\ket{1}$.
Thus, an arbitrary state in the symmetric subspace can be written as
\begin{align}
    \ket{\psi}
    =
    \sum_{m=-n/2}^{n/2} c_m
    \big|D_{n/2-m}^{(n)}\big\rangle\,\,:\,\,
    \sum_{m=-n/2}^{n/2}|c_m|^2=1 .
\end{align}
The corresponding oscillator state is then
\begin{align}
    \ket{\Psi}_{\rm osc} \eq \sum_{m=-n/2}^{n/2} c_m \ket{n/2-m}_{\rm osc} \eq T\ket{\psi}\,,
\end{align}
where $\ket{k}_{\rm osc}$ denotes the $k$-excitation eigenstate of $a^\dag a$ of the bosonic mode, and
\begin{align}
    T \eq \sum_{m=-n/2}^{n/2}
    \ket{n/2-m}_{\rm osc}
    \big\langle D_{n/2-m}^{(n)}\big|
\end{align}
is an isometry from the totally symmetric subspace to the bosonic Hilbert space.

Thus, by first preparing the desired symmetric qubit state $\ket{\psi}$ and then applying $\Vswap$, one can prepare an arbitrary bosonic state supported on the $(n+1)$-dimensional subspace spanned by the Fock states with at most $n$ excitations.
In summary, combining $\Vswap$ with \cref{cor:state_prep} provides a method for preparing the bosonic mode in an arbitrary pure state with up to $n$ excitations, e.g., phonons in the case of motional oscillators.
Precisely,
\begin{corollary}
    \textbf{(Preparing bosonic states)}
    For any state $\ket{\Psi}$ supported on the lowest $(n+1)$ energy levels of
    the bosonic mode, and any state $\ket{\phi}$ supported on the
    $(n+1)$-dimensional symmetric subspace of the qubits, e.g., $|0\rangle^{\otimes n}$ there exists a unitary
    $V\in\V_{z-x}$ such that
    \begin{align}
        V\big(\ket{\phi}\otimes\vac\big)
        =
        \ket{1}^{\otimes n}\otimes\ket{\Psi}_{\rm osc}\, .
    \end{align}
    \label{cor:state_prep2}
\end{corollary}
This state-preparation protocol is exact and consists solely of unitary time evolution under the Hamiltonians $\HTC$, $J_z$, and $J_x$.

\Cref{fig:swap_circuit} shows an explicit construction for the 2-qubit $\Vswap$ gate, with further details provided in \cref{app:state_prep_circuit}.
We note that prior works have studied other methods for preparing arbitrary states \cite{Law_Eberly_1996} and implementing arbitrary unitaries acting on finite subspaces of the bosonic mode's energy levels \cite{Santos_2005,Strauch_2012,Mischuck_Molmer_2013,Heeres_etal_2015,Krastanov_etal_2015,Fosel_etal_2020,Liu_etal_2021}, using a single-qubit coupling.
For example, \cite{Liu_etal_2021} describes a method for implementing arbitrary unitaries on a subspace of the oscillator with a finite number of excitations, 
using a Jaynes-Cummings interaction and $x$ field. 
(See \cite{Liu_etal_2021} for a discussion of these existing protocols.)

\begin{figure}[htp]
 \centering
 \begin{adjustbox}{width=0.7\columnwidth}
    \includegraphics[]{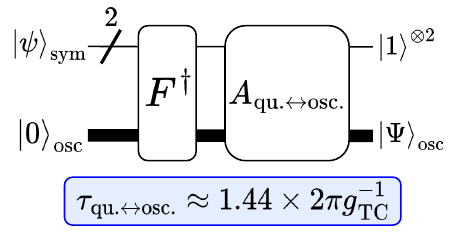}
 \end{adjustbox}
    \caption{Quantum circuit implementing the 2-qubit $\Vswap$ gate. The $A_{\text{qu.}\leftrightarrow\text{osc.}}$-gate is implemented using a circuit of the form given for $A(U)$ in \cref{tab:gadget_circuits}, with parameters listed in \cref{tab:A_gate_params}.}
    \label{fig:swap_circuit}
\end{figure}

\section{Summary \& Discussion}
In this paper, we studied the realization of permutation-invariant multi-qubit gates using the Tavis-Cummings interaction.
In particular, we showed that \textit{all} $n$-qubit unitaries respecting permutation symmetry can be realized by coupling the qubits to a bosonic mode initialized in and returned to its vacuum state, using the TC interaction together with global $z$ and $x$ fields that act identically on all qubits (\cref{thm:mainThm}).
This controllability result provides a powerful framework for realizing collective entangling operations, which are widely used in quantum computing.
An immediate and useful consequence is that these interactions can be used to prepare arbitrary pure states within the symmetric subspace of the qubits.
Combined with the state-transfer gadget introduced in our companion paper, this also leads to a natural protocol for preparing arbitrary states within the lowest $(n+1)$ energy levels of the bosonic mode.

PI $n$-qubit gates arise naturally in various applications, including variational quantum machine learning \cite{meyer2023exploiting, nguyen2022theory, sauvage2022building, zheng2023speeding} and symmetry-preserving variational quantum eigensolvers for quantum chemistry \cite{barron2021preserving}.
Our results show that such gates need not be synthesized from generic two-qubit building blocks, but can instead be implemented directly using a collective qubit-oscillator interaction.
This is especially relevant in the presence of decoherence, where direct implementations of many-body gates can be substantially more efficient than decompositions into universal two-qubit gates.
Related direct implementations of $n$-body gates have recently been developed using trapped ions \cite{Katz_etal_2022,Katz_etal_2023,Haffner_etal_2005} and atoms in optical cavities \cite{Goto_Ichimura_2004,Luo_etal_2024_3-4body}.

In the special case of $n=2$ qubits, we presented explicit quantum circuits realizing \textit{all} PI U(1)-invariant unitaries using only the TC interaction and a global $z$ field; these circuits consist of alternating sequences of unitaries generated by the Hamiltonians $\HTC$ and $J_z$.
The total time for which the TC interaction is turned on is bounded and is on the order of $2\pi\gTC^{-1}$, the inverse coupling strength.
We also provided explicit examples of circuits implementing useful entangling gates, including CZ, SWAP, iSWAP, and $\sqrt{\mathrm{iSWAP}}$.

Since these methods require only collective control over the qubits, they are potentially easier to implement and more efficient than traditional decompositions of multi-qubit gates into fixed entangling gates, such as CNOT, which break permutation symmetry.
More broadly, these results suggest that bosonic modes can serve not only as ancillas for few-qubit gate synthesis, but also as useful resources for collective, symmetry-adapted quantum operations.

\section*{Acknowledgments}
This work was supported in part by a collaboration between the U.S. Department of Energy and other agencies.
This material is based upon work supported by the U.S. Department of Energy, Office of Science, National Quantum Information Science Research Centers, Quantum Systems Accelerator (Award No. DE-SCL0000121).
Additional support is acknowledged from DOE Office of Advanced Scientific Computing Research, under Award No. DE-SC0026321, NSF PHY-2046195, NSF QLCI grant OMA-2120757, and ARL-ARO QCISS Grant number W911NF-21-1-0005.
We thank David Jakab for helpful discussions during the early stages of this project.

\bibliography{main_bib}

\newpage 
\onecolumngrid
\appendix

\newpage
\newcommand\appitemtwo[2]{{\textbf{\cref{#1} \nameref*{#1}}} \dotfill \pageref{#1}\\ \begin{minipage}[t]{0.9\textwidth} #2\end{minipage}}
\newcommand\appitem[1]{\hyperref[{#1}]{\textbf{\cref{#1}}} \textbf{\nameref*{#1}}
\dotfill \pageref{#1}\vspace{5pt}}
\newcommand\subappitem[1]{{\textbf{\ref{#1}.} \nameref*{#1}} \hspace*{\fill} \pageref{#1}}
\makeatletter
\newcommand{\appsec}[2]{
  \section{#1}
  \def\@currentlabelname{#1}
  \def\@currentlabel{\thesection}
  \label{#2}
  \addcontentsline{toc}{section}{#1}
}
\newcommand{\appsubsec}[2]{
  \subsection{#1}
  \def\@currentlabelname{#1}
  \def\@currentlabel{\thesubsection}
  \label{#2}
}
\newcommand{\appsubsubsec}[2]{
  \refstepcounter{subsubsection}
  \subsubsection{#1}
  \addtocounter{subsubsection}{-1}
  \def\@currentlabelname{#1}
  \def\@currentlabel{\thesubsubsection}
  \label{#2}
}
\makeatother

\section*{Appendix: Table of Contents}
\begin{itemize}[label={}]
\item \appitem{app:schur_weyl_basis}
\item \appitem{app:theorem_proof}
\item \appitem{app:comp_basis_diag}
\item \appitem{app:two_qubit_circuits}
\subitem \subappitem{app:2q_circuits_basics}
\subitem \subappitem{app:A_gate}
\subitem \subappitem{app:Agate_optimization}
\subitem \subappitem{app:minimizing_2qubit_unis}
\subitem \subappitem{app:useful_gate_table}
\subitem \subappitem{app:A_gate_params}
\subitem \subappitem{app:F_gate}
\subitem \subappitem{app:F_gate_geometry}
\item \appitem{app:state_prep_circuit}
\end{itemize}

\newpage
\appsec{Schur-Weyl duality \& a basis for the combined qubit-bosonic Hilbert space}{app:schur_weyl_basis}
First, we briefly discuss the decomposition of the $n$-qubit Hilbert space into irreducible representations (irreps) of SU(2) and the symmetric group $\Sn$, and we define a suitable basis for the combined qubit-bosonic Hilbert space.
Consider a system of $n$ qubits, which is described by the Hilbert space $(\C^2)^{\otimes n}$.
Single qubit rotations acting on all qubits identically via
\begin{align}
    U\otimes U\otimes\dotsc\otimes U \eq U^{\otimes n}\,\,:\,\,U\in\SU(2)
\end{align}
define a natural representation of SU(2) on $(\C^2)^{\otimes n}$.
Furthermore, this action is PI; it commutes with the group generated by all 2-qubit SWAP unitaries,
\begin{align}
    \Big\langle\textbf{P}_{i,j}\,:\,i,j\in\{1,\,\dotsc,\,n\} \Big\rangle\,,
\end{align}
which defines a representation of $\Sn$ on $(\C^2)^{\otimes n}$.

\vspace{0.5em}
The Schur-Weyl duality \cite{schur_1927,weyl_1939, goodman2009symmetry,harrow_2005} is the fact that there is a canonical one-to-one correspondence between irreps of SU(2) and irreps of $\Sn$ appearing in $(\C^2)^{\otimes n}$.
In particular, this Hilbert space decomposes as
\begin{align} \label{eq:schur_weyl_decomp}
    (\C^2)^{\otimes n} \,\,\cong\,\, \bigoplus_{j=j_{\text{min}}}^{n/2} \big(\C^{M(n,j)}\otimes\C^{2j+1}\big)\,,
\end{align}
where each $\C^{2j+1}$ (for different $j$) factor is an SU(2) irrep appearing with multiplicity $M(n,j)$, and each $\C^{M(n,j)}$ factor is an irrep of $\Sn$ appearing with multiplicity $2j+1$.
Every PI operator is not only block-diagonal with respect to this decomposition, but acts only on the SU(2) irrep $\C^{2j+1}$.
Therefore, the group of all PI $n$-qubit unitaries is
\begin{align} \label{eq:n_qubit_PI}
    \Big\langle U_{\text{PI}}\in\mathcal{L}\big((\mathbb{C}^2)^{\otimes n}\big)\,:\,U^{\dag}U=UU^{\dag}=\1,\,\,[U_{\text{PI}},\textbf{P}_{i,j}]=0\quad\forall\,i\neq j\in\{1,\,\dotsc,\,n\}\Big\rangle \,\,\cong\,\,\bigoplus_{j=j_{\text{min}}}^{n/2} \1_{M(n,j)}\otimes\UU(2j+1)\,,
\end{align}
where $\mathcal{L}\big((\mathbb{C}^2)^{\otimes n}\big)$ is the space of linear maps on $(\C^{2})^{\otimes n}$, and $\UU(2j+1)$ is the unitary group acting on $\C^{2j+1}$.

\vspace{1em}
It is useful to use a basis for the combined qubit-bosonic Hilbert space, $(\C^2)^{\otimes n}\otimes \mathcal{L}^2(\R)$, that respects the block structure in \cref{eq:schur_weyl_decomp}.
Namely, consider the simultaneous eigenstates of $J^2$, $J_z$, and $a^{\dag}a$:
\begin{align}
    \ket{j,m,\al}\otimes\kosc \qquad\begin{cases}j&=\jmin,\dotsc,\frac{n}{2}\\[3pt]
    m&=-j,\dotsc,j\\[3pt]
    \al&=1,\dotsc,M(n,j)\\[3pt]
    k&=0,\dotsc,\infty
    \end{cases}\,,
 \label{eq:basis1}
\end{align}
where
\bes
\begin{align}
    J^2\ket{j,m,\alpha}&=j(j+1) \ket{j,m,\alpha}\\[3pt] 
    J_z \ket{j,m,\alpha}&=m \ket{j,m,\alpha}\\[3pt]
    a^{\dagger}a\kosc &= k\kosc\,.
\end{align}
\label{eq:basis2}
\ees
With respect to the $\ket{j,m,\al}$ basis for $\H_{\text{qubits}}=(\C^2)^{\otimes n}$, we can define the projector to the common eigen-subspace of $J^2$ and $J_z$ within with corresponding eigenvalues $j(j+1)$ and $m$,
\begin{align} \label{eq:Pjm}
    P_{j,m} \,:=\, \sum_{\al}\pure{j,m,\al}\,.
\end{align}
With respect to the decomposition in \cref{eq:schur_weyl_decomp}, we denote this projector as
\begin{align} \label{eq:Pjm2}
    P_{j,m} \eq \1_{M(n,j)}\otimes\pure{j,m}\,,
\end{align}
where $\ket{j,m}$ denotes the shared eigenstate of $J^2$ of $J_z$, with corresponding eigenvalues $j(j+1)$ and $m$, within a fixed but arbitrary SU(2) irrep labeled by $j$.

\subsection{PI, U(1)-invariant $n$-qubit unitaries}
In \cref{sec:TC_and_z}, we claimed that the group of $n$-qubit PI, $J_z$-conserving unitaries is precisely the set of all unitaries in the form
\begin{align}
    U \eq \sum_{j=j_{\min}}^{n/2}\sum_{m=-j}^je^{i\phi_{j,m}} P_{j,m}\quad:\quad\phi_{j,m}\in[0,2\pi)\,,
 \label{eq:charU_app}
\end{align}
where $P_{j,m}$ is the projector defined above.
This is a consequence of Schur-Weyl duality; which as we explained in the previous section (see \cref{eq:n_qubit_PI}), implies that that an arbitrary PI unitary $U_{\text{PI}}$ is in the form
\begin{align}
    U_{\text{PI}} \eq \bigoplus_{j=j_{\text{min}}}^{n/2} \1_{M(n,j)}\otimes(U_{\text{PI}})_j\,,
\end{align}
where $(U_{\text{PI}})_j$ is a unitary acting on the SU(2) irrep $\C^{2j+1}$.
Then, if $U_{\text{PI}}$ also commutes with $J_z$, which is non-degenerate within each SU(2) irrep, it is also diagonal with respect to the eigenbasis of $J_z$ defined in \Cref{eq:basis1,eq:basis2}.
In other words, $(U_{\text{PI}})_j$ commutes with $\pure{j,m}$ for each $m=-j,\dotsc,\,j$.
This means the components in each $j$ sector are in the form
\begin{align}
    \1_{M(n,j)}\otimes(U_{\text{PI}})_j \eq \sum_{m=-j}^{j} e^{i\phi_{j,m}}P_{j,m} \quad:\quad\phi_{j,m}\in[0,2\pi)\,,
\end{align}
from which the form in \cref{eq:charU_app} follows immediately.

\newpage
\appsec{Proof of \cref{thm:mainThm}}{app:theorem_proof} Here, we prove \cref{thm:mainThm} using \cref{thm:general_k_ancilla}, which was adapted from our companion paper \cite{theory_paper}, together with the following lemma.

\begin{lemma}
Any PI unitary $U$ on $n$ qubits can be realized as a finite sequence of global rotations around the $x$-axis and PI unitaries that also conserve $J_z$.
\label{lem:PI_with_Jx}
\end{lemma}
\begin{proof}
Recall that the group of PI $n$-qubit unitaries is
\begin{align}
    \Big\{
    U\in \mathcal{L}\big((\mathbb{C}^2)^{\otimes n}\big)\,:\,U^{\dag}U=UU^{\dag}=\1,\,\,
    [U,\mathbf{P}_{i,j}]=0
    \quad \forall i\neq j
    \Big\}.
\end{align}
Using the Schur-Weyl decomposition in \cref{eq:schur_weyl_decomp}, its Lie algebra is
\begin{align}
    \mathfrak{h}
    \,:=&\,\,
    \left\{
    A\in \mathcal{L}\big((\mathbb{C}^2)^{\otimes n}\big):
    A+A^\dag=0,\,
    [A,\mathbf{P}_{i,j}]=0
    \quad \forall i\neq j
    \right\} \notag\\[6pt]
    \cong&\,\,
    \bigoplus_{j=\jmin}^{n/2}
    \1_{M(n,j)}\otimes \mathfrak{u}(2j+1).
\end{align}
Here $M(n,j)$ denotes the multiplicity space of the spin-$j$ irrep, and $\mathfrak{u}(2j+1)$ denotes the Lie algebra of skew-Hermitian operators on the spin-$j$ irrep.

The subgroup of PI unitaries that also conserve $J_z$ consists of the diagonal unitaries
\begin{align}
    U
    =
    \sum_{j=\jmin}^{n/2}
    \sum_{m=-j}^{j}
    e^{i\phi_{j,m}} P_{j,m}\quad:\quad
    \phi_{j,m}\in[0,2\pi)\,,
\end{align}
where $P_{j,m}$ is defined in \cref{eq:Pjm,eq:Pjm2}.
Equivalently, its Lie algebra is
\begin{align}
   \mathfrak{h}_z := \Big\{A\in\mathfrak{h}:[A,J_z]=0\Big\} \eq 
   \s_{\mathbb{R}}\left\{iP_{j,m}\,:\, j=\jmin,\dotsc,\frac n2,\,m=-j,\dotsc,j\right\}\,.
\end{align}
In particular, $\mathfrak{h}_z$ contains $iJ_z$.

We now show that
\begin{align}
    \mathfrak{h} \eq \alg_{\R}\big\{\mathfrak{h}_z,\,iJ_x\big\}\,.
\end{align}
Since $iJ_z\in\mathfrak{h}_z$, adding $iJ_x$ also gives
\begin{align}
    iJ_y=[iJ_z,iJ_x]\,.
\end{align}
On each spin-$j$ irrep, the operators $J_x$ and $J_y$ have the form
\begin{align}
    iJ_x &\eq i\sum_{j=\jmin}^{n/2}\1_{M(n,j)}\otimes\sum_{m=-j}^{j-1}a_{j,m}\big(\ket{j,m+1}\bra{j,m}+\ket{j,m}\bra{j,m+1}\big)\,,\\[4pt]
    iJ_y&\eq\sum_{j=\jmin}^{n/2}\1_{M(n,j)}\otimes\sum_{m=-j}^{j-1}a_{j,m}\big(\ket{j,m+1}\bra{j,m}-\ket{j,m}\bra{j,m+1}\big),
\end{align}
where
\begin{align}
    a_{j,m} := \frac12\left\langle j,m+1 \middle| J_+ \middle| j,m \right\rangle \eq \frac12\sqrt{(j-m)(j+m+1)}\,.
\end{align}
For all $m=-j,\dotsc,j-1$, these coefficients are nonzero.

Because $\mathfrak{h}_z$ contains the diagonal projectors $iP_{j,m}$, commutators with these projectors isolate the nearest-neighbor transitions
\begin{align}
    \1_{M(n,j)}\otimes \big(\ket{j,m+1}\bra{j,m} - \ket{j,m}\bra{j,m+1} \big),
\end{align}
and
\begin{align}
    i\1_{M(n,j)}\otimes\big(\ket{j,m+1}\bra{j,m} + \ket{j,m}\bra{j,m+1}\big)\,,
\end{align}
for each fixed $j$ and each $m=-j,\dotsc,j-1$.
Thus, for every spin sector $j$, the generated Lie algebra contains all diagonal skew-Hermitian matrices as well as all nearest-neighbor skew-Hermitian matrices on the chain
\[ \ket{j,-j},\,\ket{j,-j+1},\,\dotsc,\,\ket{j,j}\,.\]
It can be easily seen that since this chain is connected, these nearest-neighbor generators generate all matrix units on the spin-$j$ irrep.
Hence they generate
\begin{align}
    \1_{M(n,j)}\otimes\mathfrak{u}(2j+1)
\end{align}
for each $j$.
Therefore,
\begin{align}
    \alg_{\R}\big\{\mathfrak{h}_z,\,iJ_x\big\} \eq 
    \bigoplus_{j=\jmin}^{n/2} \1_{M(n,j)}\otimes\mathfrak{u}(2j+1) \eq \mathfrak{h}.
\end{align}

It follows that the connected Lie group generated by the two subgroups
\[
    \exp(i \theta J_x): \theta\in[0,4\pi)\,
    \quad\text{and}\quad\,
    \exp(\mathfrak{h}_z)
\]
is the full group of PI unitaries.
Since this group is connected, every PI unitary can therefore be written as a finite product of global $x$ rotations and PI unitaries that commute with $J_z$. This proves the lemma.
\end{proof}

\vspace{4mm}
Using \cref{thm:general_k_ancilla} and \cref{lem:PI_with_Jx}, we can now prove:\\\\
\textbf{Theorem\,\,\ref{thm:mainThm}.} All PI unitary transformations on $n$ qubits can be realized, up to a global phase, using the Hamiltonian $H_{z-x}(t)$ in \cref{eq:totH}, i.e., uniform global $z$ and $x$ fields $J_z$ and $J_x$ and the Tavis-Cummings interaction $H_{\text{TC}}$, where the bosonic mode (oscillator) is initialized in and returned to an arbitrary eigenstate $\ket{k}_{\text{osc}}$ of the intrinsic oscillator Hamiltonian $a^{\dag}a$.

\begin{proof}(\cref{thm:mainThm})
First, consider PI, U(1)-invariant unitaries realized using Hamiltonians $\HTC$ and $J_z$, as described in \cref{thm:general_k_ancilla}.
Then, the unitary $X^{\otimes n}$ which, up to a global phase, corresponds to a $\pi$ rotation around $x$, is sufficient to overcome the constraints in \cref{eq:const0_main}, so that up to a global phase, \textit{all} PI, U(1)-invariant unitaries can be realized using some $V\in\V_z$ and $X^{\otimes n}$. 
In particular, for any choice of phases $\ga_j\in[0,2\pi)$, consider the unitary
\begin{align}\label{eq:X_unitary}
    X^{\otimes n}\Bigg(\mathbb{I}^{\otimes n} +\sum_{j=j_{\text{min}}}^{n/2}  \big(e^{i\ga_j}-1\big) P_{j,j}\Bigg) X^{\otimes n}\,.
\end{align}
According to \cref{thm:general_k_ancilla}, the unitary in parentheses is realizable using Hamiltonians $\HTC$ and $J_z$, provided that $\ga_0=0$ in the case of even $n$.
Since $X^{\otimes n} P_{j,m} X^{\otimes n}=P_{j,-m}$, the overall unitary in \cref{eq:X_unitary} can be written as a unitary in the form of \cref{eq:charU}, with $\phi_{j,-j}=\ga_j$ and $\phi_{j,m}=0$ for $m\neq -j$. 
Combining such unitaries with those allowed by \cref{thm:general_k_ancilla}, we obtain \textit{all} unitaries of the form in \cref{eq:charU}, i.e. \textit{all} PI, U(1)-invariant unitaries that satisfy the constraint $\phi_{0,0} = 0$ in the case of even $n$. 
In this way, we obtain \textit{all} PI, U(1)-invariant unitaries, up to a global phase.
Finally, \cref{lem:PI_with_Jx} immediately implies that PI, U(1)-invariant unitaries together with rotations around $x$ generate \textit{all} PI unitaries.
\end{proof}

\newpage
\appsec{Which diagonal PI unitaries are realizable using TC and global $z$ field?}{app:comp_basis_diag}
Here, we prove the following corollary, stated in \cref{sec:diag_unis}:\\

\noindent\textbf{Corollary \ref{cor:diag}.} A PI unitary $U$ that is diagonal in the computational basis is realizable using Hamiltonians $\HTC$ and $J_z$, up to a global phase $\alpha$ -- where the oscillator is initialized in and returned to its vacuum state $\vac$ -- if and only if, there exists $\alpha\in[-\pi,\pi)$ and $\beta\in[-2\pi,2\pi)$ such that,
\begin{align}
    \phi_m &=
    \begin{cases}
        \alpha+m \beta  \,\,(\mathrm{mod}\,2\pi) &\text{for}\,\,m\leq0\,, \\[10pt]
        \text{arbitrary} &\text{for}\,\,m>0\,.
    \end{cases}
 \label{eq:phi_m_app}
\end{align}
\begin{proof}
Recall from \cref{sec:diag_unis} that any PI unitary $U$ that is diagonal in the computational basis has the form
\begin{align}
    U \eq \sum_{m=-n/2}^{n/2} e^{i\phi_m} \sum_{j=|m|}^{n/2}{P_{j,m}} \quad:\quad\phi_m\in[0,2\pi)\,,
 \label{eq:djd}
\end{align}
where $P_{j,m}$ is the projector to the common eigen-subspace of $n$-qubit operators $J^2$ and $J_z$ with corresponding eigenvalues $j(j+1)$ and $m$ (see \cref{eq:Pjm}).

Now, we determine the condition on $\phi_m: m=-n/2,\,\dotsc,\,n/2$ such that the resulting unitary satisfies condition (2) in \cref{thm:general_k_ancilla}, and vice versa.
First we show that if $\phi_m$ satisfy \cref{eq:phi_m_app} for each $m=-n/2,\,\dotsc,\,n/2$, then $U$ satisfies the constraints in \cref{thm:general_k_ancilla}.
Recall that general PI, U(1)-invariant $n$-qubit unitaries are in the form
\begin{align}
    U \eq \sum_{j=j_{\min}}^{n/2}\sum_{m=-j}^je^{i\phi_{j,m}} P_{j,m}\quad:\quad\phi_{j,m}\in[0,2\pi)\,.
\end{align}
Then, comparing this with \cref{eq:djd}, we conclude that for all $m$ and $j\ge |m|$:
\begin{align}
    \phi_{j,m}=\phi_m\,.
\end{align}
Furthermore, the assumption that for $m\leq0$, $\phi_m=\alpha+m \beta\,\,(\mathrm{mod}\,2\pi)$ implies that for all $j=\jmin,\,\dotsc,\, n/2$,
\begin{align}
    \phi_{j,-j} \eq \phi_{-j} \eq \alpha-j \beta\,\,(\text{mod}\,2\pi)\,.
\end{align}
Thus, condition (2) of \cref{thm:general_k_ancilla} is satisfied, and hence the theorem implies that $U$ is realizable using $\HTC$ and $J_z$, up to a global phase.

Conversely, suppose that $U$ is an arbitrary PI diagonal unitary realizable using the Hamiltonians $\HTC$ and $J_z$, up to a global phase. 
By \cref{thm:general_k_ancilla}, the latter assumption implies that  $\phi_{j,-j}=\alpha+j \beta'\,\,(\text{mod}\,2\pi)$ for some $\al\in[-\pi,\pi)$ and $\beta'\in[-2\pi,2\pi)$.
Furthermore, by assumption, $U$ is diagonal and PI, and therefore admits a decomposition of the form \cref{eq:djd}.
For each $-n/2\leq m\leq 0$, the projector
\[ P_m=\sum_{j=-m}^{n/2}P_{j,m} \]
contains $P_{-m,m}$. It follows that if $m$ is in the interval $-n/2\leq m\leq 0$, then
\begin{align}
    \phi_m \eq \phi_{-m,m} \eq \alpha-m\beta'\,\,(\mathrm{mod}\,2\pi),
\end{align}
which agrees with \cref{eq:phi_m_app} after setting $\beta=-\beta'$.

For $m>0$, however, the projector
\[ P_m=\sum_{j=m}^{n/2}P_{j,m} \]
contains no projector of the form $P_{j,-j}$.
Thus, \cref{thm:general_k_ancilla} imposes no constraint on the corresponding phase $e^{i\phi_m}$, so this phase is unrestricted.
\end{proof}

\newpage
\appsec{Two-Qubit Circuit Constructions}{app:two_qubit_circuits}
\appsubsec{Basics}{app:2q_circuits_basics}
In this section, we study the dynamics of a pair of qubits, coupled to a bosonic mode initialized in the vacuum state $\vac$, under Hamiltonians $\HTC$ and $J_z$.
That is, we consider general initial states $\ket{\psi}\otimes \vac$, where $\ket{\psi}\in\mathbb{C}^2\otimes \mathbb{C}^2$ is an arbitrary state of $n=2$ qubits.  

For $n=2$ qubits, the total angular momentum $j$ takes values $0$ and $1$, each appearing with multiplicity one.
The associated multiplicity spaces correspond to the symmetric and anti-symmetric representations of $\mathbb{S}_2$, which are both 1D.
In particular,
\begin{align}
    \mathbb{C}^2\otimes \mathbb{C}^2 \eq \mathbb{C} \ket{\Psi^-} \oplus \mathcal{H}_{\text{sym}} \eq \mathbb{C} \ket{\Psi^-} \oplus \text{Span}_{\mathbb{C}}\big\{\ket{\Psi^+}\,,
    \ket{00}\,,\ket{11} \big\}\,,
\end{align}
where
\begin{align}
 \begin{split}
    \ket{j=0,m=0}& \eq \frac{\ket{01}-\ket{10}}{\sqrt{2}} \eq \ket{\Psi^-}\,,\\[4pt]
    \ket{j=1,m=0}& \eq \frac{\ket{01}+\ket{10}}{\sqrt{2}} \eq \ket{\Psi^+}\,,\\[4pt]
    \ket{j=1,m=1}& \,=:\, \ket{00}\,,\\[10pt]
    \ket{j=1,m=-1}& \,=:\, \ket{11}\,,
 \end{split}
\end{align}
the right-hand side written in the computational basis, defined by $\sigma_z\ket{0}=\ket{0}$ and $\sigma_z\ket{1}=-\ket{1}$. 

State $\ket{\Psi^-}\otimes\vac$ is an eigenvector with eigenvalue 0 of both $\HTC$ and $J_z$, so it evolves trivially under these Hamiltonians.
Furthermore, since operator $Q=a^\dag a+J_z+n/2$ is conserved by both Hamiltonians, and because for any initial state $\ket{\psi}\otimes \vac$, the time-evolved state is restricted to the eigen-subspaces with eigenvalues $q\leq 2$, it follows that any state initially restricted to $\H_{\text{sym}}\otimes\vac$ remains restricted to the subspace
\begin{align}
 \begin{split}
    \bigoplus_{q=0}^2 \H_{q,\,j=1} \, \,\,&=\,\,  \mathbb{C}\ket{11}\otimes \vac\\ &\qquad \oplus 
    \text{Span}_{\mathbb{C}}\Big\{\ket{\Psi^+} \otimes \vac\, , \ket{11}\otimes \ket{1}_{\text{osc}}\Big\}\\[6pt]
    &\qquad \oplus \text{Span}_{\mathbb{C}} \,\Big\{\ket{00}\otimes \vac\,,\,\ket{\Psi^+}\otimes \ket{1}_{\text{osc}}\,,\,\ket{11}\otimes \ket{2}_{\text{osc}}\Big\} \\[10pt] 
    &\cong\, \C\oplus\C^2\oplus\C^3\,.
 \end{split}
 \label{eq:2q_subspace}
\end{align}
Additionally, since $\HTC$ and $J_z$ respect the permutational and U(1) symmetries, they are block-diagonal with respect to this decomposition.

\begin{figure*}
    \centering
    \includegraphics[width=0.7\textwidth]{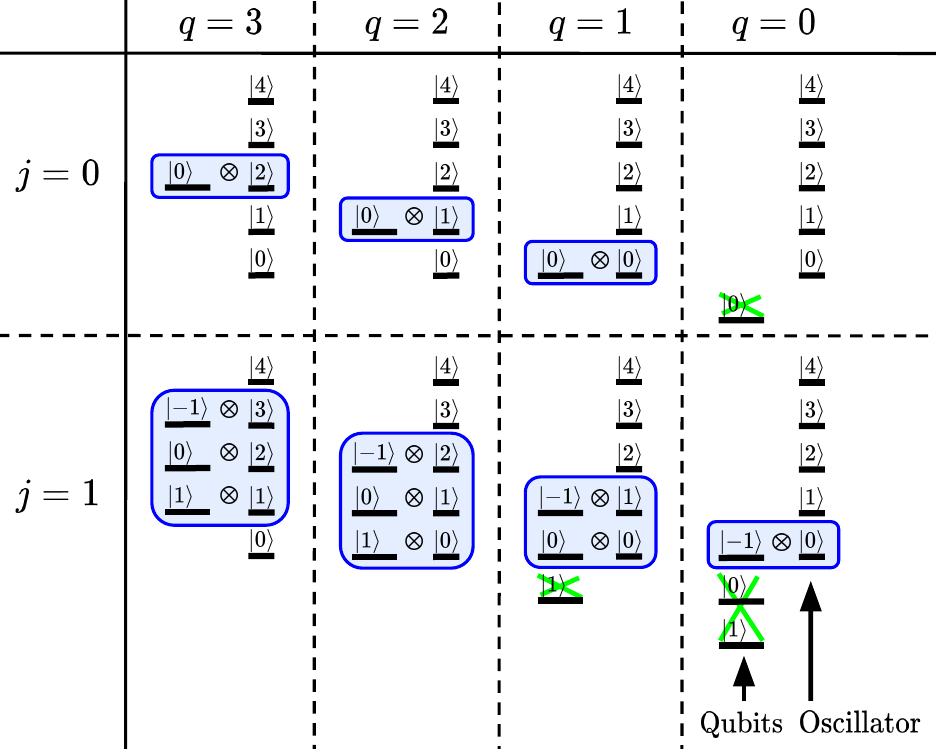}
    \caption{(from the companion paper \cite{symmetry_paper}) \textbf{2-qubit example:} $\H\qj$ is the span of states contained in the corresponding blue box. Note how the dimension of $\H\qj$ saturates at $2j+1$, for $q\geq j+\frac{n}{2}$.}
    \label{fig:2q_sectors}
\end{figure*}

Therefore, if the bosonic mode is initialized in its vacuum state $\vac$, it is sufficient to consider the action of $\V_z$ on the subspaces in \cref{eq:2q_subspace}.

It is natural to work with the matrices of $\HTC$, $\Zhat$, and $\Nhat$ in the \textit{ordered} basis described by \cref{eq:2q_subspace}, that is
\begin{align}
    \Big(\ket{11}\otimes\vac\,,\,\,\ket{\Psi^+}\otimes\vac\,,\,\,\ket{11}\otimes\ket{1}_{\text{osc}}\,,\,\,\ket{00}\otimes\vac\,,\,\,\ket{\Psi^+}\otimes\ket{1}_{\text{osc}}\,,\,\,\ket{11}\otimes\ket{2}_{\text{osc}}\Big)\,.
\end{align}
With respect to this basis, the matrices of $\HTC$, $J_z$, $a^{\dag}a$, and $Q$ are
\begin{align}
    \left[\HTC\right] &\eq 
    \arraycolsep=4pt\def\arraystretch{1.4}\left(
    \begin{array}{c|cc|ccc}
        0 &&&&&\\ \hline
        & 0 & \sqrt{2} &&& \\
        & \sqrt{2} & 0 &&& \\ \hline
        &&& 0 & \sqrt{2} & 0 \\
        &&& \sqrt{2} & 0 & 2 \\
        &&& 0 & 2 & 0
    \end{array}
    \right) \,\,,\qquad
    \left[\Zhat\right] \eq 
    \arraycolsep=2pt\def\arraystretch{1.1}\left(
    \begin{array}{c|cc|ccc}
        -1 &&&&&\\ \hline
        & 0 &&&& \\
        && -1 &&& \\ \hline
        &&& 1 && \\
        &&&& 0 &\\
        &&&&& -1
    \end{array}
    \right)\,, \\[8pt]
    \left[a^{\dagger}a\right] &\eq 
    \arraycolsep=4pt\def\arraystretch{1.1}\left(
    \begin{array}{c|cc|ccc}
        0 &&&&&\\ \hline
        & 0 &&&& \\
        && 1 &&& \\ \hline
        &&& 0 && \\
        &&&& 1 &\\
        &&&&& 2
    \end{array}
    \right) \,\,,\qquad
    \left[Q\right]\eq\left[a^{\dagger}a+\Zhat+\1\right] \eq 
    \arraycolsep=4pt\def\arraystretch{1.1}\left(
    \begin{array}{c|cc|ccc}
        0 &&&&&\\ \hline
        & 1 &&&& \\
        && 1 &&& \\ \hline
        &&& 2 && \\
        &&&& 2 &\\
        &&&&& 2
    \end{array}
    \right)\,.
 \label{eq:HJC_2qubit}
\end{align}
The block structure is indicated explicitly, with $q=0$ in the top left and $q=2$ in the bottom right.

The first statement of \Cref{thm:23_qubits} says that for any $V_1\in\SU(2)$ and $V_2\in\SU(3)$, the group of unitaries $\V_z$ realized by Hamiltonians $\HTC$ and $\Zhat$ contains
\begin{align*}
    V \eq 1\oplus V_1 \oplus V_2\,,
\end{align*}
with respect to the decomposition in \cref{eq:2q_subspace}.

It is also useful to write the matrices of the unitary operators $e^{i\HTC t}$ and $e^{i\Zhat t}$:
\begin{align}
    e^{i\HTC t} &\eq 
    \arraycolsep=4pt\def\arraystretch{1.4}
    \left(
    \begin{array}{c|cc|ccc}
        1 &&&&&\\ \hline
        & \cos(\sqrt{2}t) & i\sin(\sqrt{2}t) &&& \\
        & i\sin(\sqrt{2}t) & \cos(\sqrt{2}t) &&& \\ \hline
        &&& \frac{1}{3}\left[\cos(\sqrt{6}t)+2\right]&\frac{i}{\sqrt{3}}\sin(\sqrt{6}t)&\frac{\sqrt{2}}{3}\left[\cos(\sqrt{6}t)-1\right]\\ &&& \frac{i}{\sqrt{3}}\sin(\sqrt{6}t)&\cos(\sqrt{6}t)&i\sqrt{\frac{2}{3}}\sin(\sqrt{6}t)\\ &&& \frac{\sqrt{2}}{3}\left[\cos(\sqrt{6}t)-1\right]&i\sqrt{\frac{2}{3}}\sin(\sqrt{6}t)&\frac{1}{3}\left[2\cos(\sqrt{6}t)+1\right]
    \end{array}
    \right)
    \\[12pt]
    e^{i\Zhat t} &\eq
    \arraycolsep=4pt\def\arraystretch{1.4}
    \left(
    \begin{array}{c|cc|ccc}
        e^{-it} &&&&&\\ \hline
        & 1 & 0 &&& \\
        & 0 & e^{-it} &&& \\ \hline
        &&& e^{it} && \\ &&& & 1 & \\ &&&  && e^{-it}
    \end{array}
    \right)
    \,.
 \label{eq:HJC_gate}
\end{align}

\appsubsec{$A$-gate}{app:A_gate}
The goal of this section is to construct unitaries of the form
\begin{align}
    A\big(U\big) := 1\oplus U \oplus \1_2\,,\quad\text{where}\,\,U\in\SU(\H_{1,1})\simeq\SU(2)\,\,\text{is arbitrary.}
 \label{eq:Agate_definition}
\end{align}
This gate is used in all of our two-qubit circuit constructions from \cref{sec:useful_2q_gate}.

First, consider the fixed gate
\begin{align}
    e^{-i\HTC2\pi/\sqrt{6}} \eq 1\,\oplus\, R_x(4\pi/\sqrt{3})\,\oplus\, \1_2 \qeq \arraycolsep=6pt\def\arraystretch{1.8}
    \left(
    \begin{array}{c|c|c}
        1 && \\ \hline
        & R_x(4\pi/\sqrt{3}) & \\ \hline 
        && \1
    \end{array}
    \right)\,,
\end{align}
where $R_x(4\pi/\sqrt{3}):=e^{-i\sigma_x2\pi/\sqrt{3}}$ denotes a rotation by angle $4\pi/\sqrt{3}$ about the $x$-axis on the Bloch sphere.
Conjugating this gate by a well-chosen block-diagonal unitary $V$ which acts as $V_1$ in the $q=1$ sector, yields in the $q=1$ sector a rotation of fixed angle $4\pi/\sqrt{3}$ about an \textit{arbitrary} axis $\hat{n}$, and in $q=0$ and $q=2$ sectors the identity:
\begin{align}
    V\,e^{-i\HTC2\pi/\sqrt{6}}\,V^{\dag} \qeq \arraycolsep=6pt\def\arraystretch{1.8}\left(
    \begin{array}{c|c|c}
        1 && \\ \hline
        & V_1R_{x}(4\pi/\sqrt{3})V_1^{\dag} & \\ \hline 
        && \1
    \end{array}
    \right)
    \eq \arraycolsep=6pt\def\arraystretch{1.8}
    \left(
    \begin{array}{c|c|c}
        1 && \\ \hline
        & R_{\hat{n}}(4\pi/\sqrt{3}) & \\ \hline 
        && \1
    \end{array}
    \right)\,.
 \label{eq:Rn_gate0}
\end{align}
Such a $V_1$ can be realized using an Euler decomposition, by noting that in the $q=1$ sector, $\HTC$ looks like $\sigma_x$, and $\Zhat$ looks like $\sigma_z$.
In particular,
\begin{align}
    \pi_{1,1}(\HTC) \eq \sqrt{2}\sigma_x\,\,,\quad \pi_{1,1}(\Zhat) \eq \frac{1}{2}\Big(\sigma_z-\mathds{1}\Big)\,,
 \label{eq:charge1_pauli}
\end{align}
or
\begin{align}
    \sigma_x \eq \frac{\pi_{1,1}(\HTC)}{\sqrt{2}}\,\,,\quad \sigma_z\eq \pi_{1,1}\Big(\mathds{1}+2\Zhat\Big)\,.
\end{align}
In fact, the decomposition $V_1=e^{i\beta\sigma_x}e^{i\ga\sigma_z}$, where $\beta,\gamma\in[-\pi,\pi]$, is sufficient for taking the $x$-axis to any arbitrary axis, so consider
\begin{align}
    V \eq e^{i\beta\HTC/\sqrt{2}}e^{i2\ga\Zhat} &\eq \arraycolsep=6pt\def\arraystretch{1.8}
    \left(
    \begin{array}{c|c|c}
        e^{-i2\gamma} && \\ \hline
        & e^{-i\ga}V_1 & \\ \hline 
        && \widetilde{V}_2
    \end{array}
    \right)\,,
 \label{eq:Agate_Euler}
\end{align}
where $\widetilde{V}_2\in\SU(3)$ is some unitary that depends on $V_1$.
Thus, for any axis $\hat{n}$, and appropriately chosen Euler angles $\beta:=\sqrt{2}\theta_2$ and $\gamma:=\theta_1/2$,
\begin{align}
    \arraycolsep=6pt\def\arraystretch{1.8}
    \left(
    \begin{array}{c|c|c}
        1 && \\ \hline
        & R_{\hat{n}}(4\pi/\sqrt{3}) & \\ \hline 
        && \1
    \end{array}
    \right) &\qeq e^{i\theta_2\HTC}e^{i\theta_1\Zhat}\,\,e^{-i\HTC2\pi/\sqrt{6}} \,\,e^{-i\theta_1\Zhat}e^{-i\theta_2\HTC}\,.
 \label{eq:Rn_gate}
\end{align}

Also, for any integer $k$, rotations by angle $4\pi k/\sqrt{3}$ can be implemented using $k$ instances of $e^{-i\HTC2\pi/\sqrt{6}}$ and the same Euler angles, that is
\begin{align}
    \arraycolsep=6pt\def\arraystretch{1.8}
    \left(
    \begin{array}{c|c|c}
        1 && \\ \hline
        & R_{\hat{n}}(4\pi k/\sqrt{3}) & \\ \hline 
        && \1
    \end{array}
    \right) &\qeq e^{i\theta_2\HTC}e^{i\theta_1\Zhat}\,\,e^{-i\HTC2\pi k/\sqrt{6}} \,\,e^{-i\theta_1\Zhat}e^{-i\theta_2\HTC}\,.
 \label{eq:Rn_gate^k}
\end{align}
Next, we want to use the $2\times2$ unitaries $\exp\big(i(\hat{n}\cdot\vec{\sigma})2k\pi/\sqrt{3}\big)$ -- i.e. rotations of fixed angle but arbitrary axis -- to realize unitaries of arbitrary angle and arbitrary axis.
The following lemmas describe how to do this, with as few steps as possible.

\begin{lemma} \label{lem:rotation_mult_lemma}
Consider the set of all SU(2) elements in the form
\begin{align}
    \exp(i \gamma_1\hat{m}\cdot\hat{\sigma}) \exp(i\gamma_{2}\hat{n}\cdot\hat{\sigma})\,,
\end{align}
where $\gamma_1,\gamma_2\in[-\pi,\pi)$ are two fixed angles, whereas $\hat{n},\hat{m}\in\mathbb{R}^3$ are arbitrary unit vectors.
For any given values of $\gamma_1$ and $\gamma_2$, the above set is equal to the set of all SU(2) elements $\exp(i\alpha\hat{r}\cdot\vec{\sigma})$, with arbitrary unit vector $\hat{r}\in\mathbb{R}^3$ and angle $\alpha$ that satisfies the condition
\begin{align}
    \cos(\ga_1) \cos(\ga_2)- |\sin(\ga_1)| \times |\sin(\ga_2)| \,\,\leq\,\, \cos(\alpha) \,\,\leq\,\, \cos(\ga_1) \cos(\ga_2) + |\sin(\ga_1)| \times |\sin(\ga_2)|\,,
 \label{eq:cosA_bound}
\end{align}
or equivalently, $|\cos(\al) -\cos(\ga_1)\cos( \ga_2)| \le |\sin(\ga_1)| \times |\sin(\ga_2)|$.
\end{lemma}
\begin{proof}
We use the fact that $\exp(i\gamma\hat{m}\cdot\hat{\sigma})=\cos(\ga) \mathbb{I}+i\sin(\ga)\hat{m}\cdot\hat{\sigma}$.
Furthermore, the anti-commutation relation $\sigma_{i}\sigma_j+\sigma_{j}\sigma_i=2\delta_{i,j}\mathbb{I}$ implies that
\begin{align}
    (\hat{m}\cdot\hat{\sigma})(\hat{n}\cdot\hat{\sigma}) \eq (\hat{m}\cdot\hat{n})\1 + (\hat{m}\times\hat{n})\cdot\vec{\sigma}\,.
\end{align}
Therefore,
\begin{align}
    \exp(i \gamma_1\hat{m}\cdot\hat{\sigma}) \exp(i\gamma_{2}\hat{n}\cdot\hat{\sigma}) &\eq \Big(c_1c_2 - s_1s_2(\hat{m}\cdot\hat{n})\Big)\1 \p i\Big(s_1c_2\hat{m}+c_1s_2\hat{n}-s_1s_2(\hat{m}\times\hat{n})\Big)\cdot\vec{\sigma}\,,
\end{align}
where $c_j:=\cos(\gamma_j)$ and $s_j:=\sin(\gamma_j)$.
Equivalently, this product can be written as a single element $\exp(i\al\hat{r}\cdot\vec{\sigma})$, where
\begin{align}
    \cos(\al) &\eq c_1c_2-s_1s_2(\hat{m}\cdot\hat{n}) \\[4pt]
    \hat{r}\sin(\al) &\eq s_1c_2\hat{m}+c_1s_2\hat{n}-s_1s_2(\hat{m}\times\hat{n})\,.
\end{align}
Therefore, varying the angle between axes $\hat{m}$ and $\hat{n}$ --- or equivalently, varying $\hat{m}\cdot\hat{n}\in[-1,1]$ --- yields \cref{eq:cosA_bound}, that is
\begin{align}
    c_1c_2-|s_1||s_2| \,\,\leq\,\, \cos(\al) \,\,\leq\,\, c_1c_2+|s_1||s_2|\,.
\end{align}
Then, even with $\hat{m}\cdot\hat{n}$ fixed, $\hat{r}$ is determined entirely by the orientation of $\hat{m}$ and $\hat{n}$ (to be precise, the direction of $\hat{m}\times\hat{n}$) -- so arbitrary $\hat{r}$ can be obtained from an appropriate choice of orientations.
\end{proof}
This lemma means that unless $|\gamma_1|=|\gamma_2|=\pi/2$, it is impossible to realize all elements of SU(2) as $\exp(i\gamma_1\hat{m}\cdot\hat{\sigma}) \exp(i\gamma_{2}\hat{n}\cdot\hat{\sigma})$.
On the other hand, the following lemmas show how general elements of SU(2) can be written as products of two or more rotations by a single fixed angle $\gamma$, about different axes.
Therefore, using the construction in \cref{eq:Rn_gate^k}, the following lemmas tell us how to implement an $A$-gate.
\begin{lemma}\label{lem:2step} Fix $\gamma\in[-\pi,\pi)$. 
For any unit vector $\hat{n}\in\R^3$, unitary $U=\exp(i \alpha \hat{n}\cdot \vec{\sigma})$ satisfying $\cos(\alpha) \ge \cos(2\gamma)$ has a decomposition as 
\begin{align}
    U= \exp(i \gamma \hat{n}_2\cdot \vec{\sigma})\exp(i\gamma \hat{n}_1\cdot \vec{\sigma})\,.
 \label{eq:2_step_unitary}
\end{align}
for some unit vectors $\hat{n}_1,\hat{n}_2\in\mathbb{R}^3$.
\end{lemma}
\begin{proof}
Follows directly from \cref{lem:rotation_mult_lemma}.
In particular, $U=\exp(i\al\hat{n}\cdot\vec{\sigma})$ has a decomposition of the form in \cref{eq:2_step_unitary} if
\begin{align}
    \cos(2\gamma) \eq \cos^2(\gamma) - \sin^2(\gamma) \,\,\leq\,\, \cos(\al) \,\,\leq\,\, 1\,.
\end{align}
\end{proof}
Note that this lemma, in particular, implies that if $\cos(\alpha) \ge \cos(2k\gamma)$ for some $k\in\mathbb{N}$, then there are unit vectors $\hat{n}_1,\hat{n}_2\in\mathbb{R}^3$ such that
\begin{align}
    U= \exp(i (k\gamma) \hat{n}_2\cdot \vec{\sigma})\exp(i (k\gamma) \hat{n}_1\cdot \vec{\sigma})\,.
 \label{eq:2k_step_unitary}
\end{align}
The number of fixed rotations by angle $2\gamma$ in this decomposition is $2k$.
For $k\geq2$, there can be more efficient decompositions than \cref{eq:2k_step_unitary}.
For example,
\begin{lemma}\label{lem:3step}
Consider a fixed $\gamma\in[-\pi,\pi)$.
Any arbitrary unitary $U\in \SU(2)$ with decomposition $U=\exp(i \alpha \hat{n}\cdot \vec{\sigma})$ for unit vector $\hat{n}\in\R^3$ and $\alpha\in[-\pi,\pi)$ has a 3-step decomposition as
\begin{align}
    U\eq\exp(i \ga [\pm \hat{n}]\cdot \vec{\sigma})  \exp(i \ga \hat{n}_2\cdot \vec{\sigma})\exp(i \ga \hat{n}_1\cdot \vec{\sigma})\,,
 \label{eq:3_step_unitary}
\end{align}
for some unit vectors $\hat{n}_1,\hat{n}_2\in\mathbb{R}^3$, if 
$\cos(\al\mp \ga) \geq\cos(2\ga)$.
\end{lemma}
\begin{proof}
This follows immediately by applying \cref{lem:2step} to unitary $U'_\pm=\exp(\pm i \ga\hat{n}\cdot \vec{\sigma}) U=\exp([\alpha\pm \gamma] i \ga\hat{n}\cdot \vec{\sigma})$.
\end{proof}

\appsubsec{Implementing an $A$-gate}{app:Agate_implementation}
The previous two lemmas imply how to efficiently implement $A$-gates using sequences of the form in \cref{eq:Rn_gate^k}.
Recall that for any unit vector $\hat{n}\in\R^3$, there are angles $\theta_1(\hat{n})\in\big[-2\pi,2\pi\big]$ and $\theta_2(\hat{n})\in\Big[-\frac{\pi}{\sqrt{2}},\frac{\pi}{\sqrt{2}}\Big]$ such that
\begin{align}
    \arraycolsep=6pt\def\arraystretch{1.8}
    \left(
    \begin{array}{c|c|c}
        1 && \\ \hline
        & \exp\big(-i(\hat{n}\cdot\vec{\sigma})2\pi k/\sqrt{3}\big) & \\ \hline 
        && \1
    \end{array}
    \right) &\qeq 
    e^{i\theta_2\HTC}e^{i\theta_1\Zhat}\,\,e^{-i\HTC2\pi k/\sqrt{6}} \,\,e^{-i\theta_1\Zhat}e^{-i\theta_2\HTC}\,.
 \label{eq:Rnk_gate}
\end{align}
\begin{proposition}
\label{prop:Agate_decomp}
Consider the fixed angle $\de:=2\pi/\sqrt{3}\approx1.15\pi$.
Then, for any unit vector $\hat{n}\in\R^3$, the unitary $U=\exp(i\al\hat{n}\cdot\vec{\sigma})\in\SU(2)$ has one of the following decompositions, based on angle $\al$:
\begin{align}
    U \eq \begin{cases}
            \exp\big(i \de\,\hat{n}_2\cdot \vec{\sigma}\big)\exp\big(i \de\,\hat{n}_1\cdot \vec{\sigma}\big) & \cos(\al)\geq\cos(2\de)\approx0.56 \\[6pt]
            \exp\big(i (2\de)\hat{n}_2\cdot \vec{\sigma}\big)\exp\big(i (2\de)\hat{n}_1\cdot \vec{\sigma}\big) & \cos(\al)\geq\cos(4\de)\approx-0.36 \\[6pt]
            \exp\big(i \de\,[\pm\hat{n}]\cdot \vec{\sigma}\big)  \exp\big(i \de\,\hat{n}_2\cdot \vec{\sigma}\big)\exp\big(i \de\,\hat{n}_1\cdot \vec{\sigma}\big) & \cos(\al)\leq\cos(3\de)\approx-0.11
          \end{cases}\,,
 \label{eq:three_cases}
\end{align}
for appropriately chosen unit vectors $\hat{n}_1,\hat{n}_2\in\R^3$ (different for each case), and at least one of $[\pm\hat{n}]$ in the third case.
The corresponding interaction times for implementing $A(U)$ are
\begin{align}
    0.82\approx\frac{2}{\sqrt{6}}\quad&\leq\quad\tau_A(\text{2-step}) \times\frac{\gTC}{2\pi}\quad\leq\quad \frac{2}{\sqrt{6}} + \frac{3}{2\sqrt{2}}\approx1.88 && \cos(\al)\geq\cos(2\de)\approx0.56 \label{eq:time_bound_2step}\\[6pt]
    1.63\approx\frac{4}{\sqrt{6}}\quad&\leq\quad\tau_A(\text{4-step}) \times\frac{\gTC}{2\pi}\quad\leq\quad \frac{4}{\sqrt{6}} + \frac{3}{2\sqrt{2}}\approx2.69 && \cos(\al)\geq\cos(4\de)\approx-0.36 \label{eq:time_bound_4step}\\[6pt]
    1.22\approx\frac{3}{\sqrt{6}}\quad&\leq\quad\tau_A(\text{3-step}) \times\frac{\gTC}{2\pi}\quad\leq\quad \frac{3}{\sqrt{6}} + \frac{4}{2\sqrt{2}}\approx2.64 && \cos(\al)\leq\cos(3\de)\approx-0.11
 \label{eq:time_bound_3step}
\end{align}
Therefore, the total interaction time required to implement an $A$-gate using this construction is bounded by
\begin{align}
    \tau_A\,\,\leq\,\,\left(\frac{4}{\sqrt{6}}+\frac{3}{2\sqrt{2}}\right)\times2\pi\gTC^{-1} \,\,\approx\,\,2.69\times2\pi\gTC^{-1}\,.
\end{align}
\end{proposition}
\begin{proof}
The first two decompositions in \cref{eq:three_cases} follow immediately from \cref{lem:2step}.
The third follows from \cref{lem:3step}. %
For any unit vector $\hat{n}\in\R^3$, consider $U=\exp(i\al\hat{n}\cdot\vec{\sigma})$, where $\cos(\al)\leq\cos(3\delta)\approx-0.11$.
Angle $\al$ is restricted to the range
\begin{align}
    0.54\pi\,\approx\,-3\de + 4\pi \,\,\leq\,\, \al \,\,\leq\,\, 3\de-2\pi\,\approx\,1.46\pi\,.
\end{align}
Split this domain in two:
\begin{align}
    [-3\de+ 4\pi,\,3\de-2\pi] &\eq [-3\de+ 4\pi,\,\de] \,\,\cup\,\, [-\de,\,3\de-2\pi] \\[2pt]
    &\,\approx\,[0.54\pi,\,1.15\pi] \,\,\cup\,\, [0.85\pi,\,1.46\pi]\,.
\end{align}
Then, for $\al\in[-3\de+ 4\pi,\,\de]$ (blue region in \cref{fig:3step_delta}),
\begin{align}
    \al+\de \in [-2\de+ 2\pi,\,2\de-2\pi] \,\approx\, [-0.31\pi,\,0.31\pi] \quad\implies\quad\cos(\al+\de) \,\geq\, \cos(2\de)\,.
\end{align}
And for $\al\in[-\de,\,3\de-2\pi]$ (red region in \cref{fig:3step_delta}),
\begin{align}
    \al-\de \in [-2\de+ 2\pi,\,2\de-2\pi] \,\approx\, [-0.31\pi,\,0.31\pi] \quad\implies\quad\cos(\al-\de) \,\geq\, \cos(2\de)\,.
\end{align}
Thus, for any $\al$ satisfying $\cos(\al)\leq\cos(3\de)$, either $\cos(\al+\de)\geq\cos(2\de)$ or $\cos(\al-\de)\geq\cos(2\de)$, so the condition of \cref{lem:3step} is satisfied.
This proof is illustrated geometrically in \cref{fig:3step_delta}.
\begin{figure}[htp]
    \centering
    \includegraphics[width=0.4\textwidth]{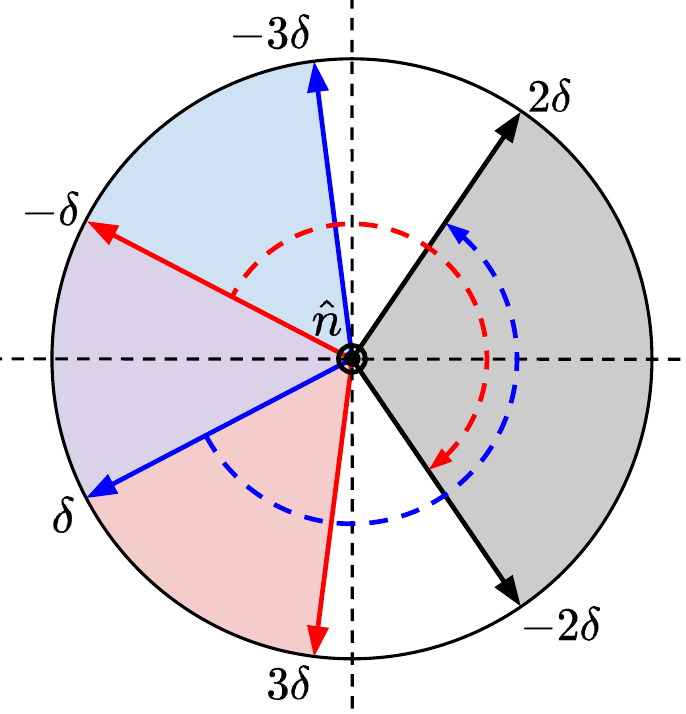}
    \caption{\textbf{Geometry of 3-step decomposition in \cref{eq:three_cases}:} (1) For $\al$ in the gray zone, i.e. $\cos(\al)\geq\cos(2\de)$, there is a $2$-step decomposition $\exp(i\al\hat{n}\cdot\vec{\sigma})=\exp(i\de\hat{n}_1\cdot\vec{\sigma})\exp(i\de\hat{n}_2\cdot\vec{\sigma})$, (2) For $\al$ in the blue zone $[-3\de+4\pi,\,\de]$, $\al+\de$ is in the gray zone, (3) For $\al$ in the red zone $[-\de,\,3\de-2\pi]$, $\al-\de$ is in the gray zone. Thus, for $\al$ in the blue \textit{or} red zones, i.e. $\cos(\al)\leq\cos(3\de)$, there is a 3-step decomposition $\exp(i\al\hat{n}\cdot\vec{\sigma})=\exp(i\de[\pm\hat{n}]\cdot\vec{\sigma})\exp(i\de\hat{n}_1\cdot\vec{\sigma})\exp(i\de\hat{n}_2\cdot\vec{\sigma})$.}
    \label{fig:3step_delta}
\end{figure}

The bounds on total interaction time in \cref{eq:time_bound_2step,eq:time_bound_4step,eq:time_bound_3step} are deduced by writing out explicitly the corresponding $A$-gates.
Using \cref{eq:Rnk_gate}, the 3-step circuit of \cref{eq:3_step_unitary} with $\gamma=\delta$ reads as:
\begin{align}
 \begin{split}
   A(U)
   &\eq \UTC\big(-\theta_2(\hat{n})\big)\,\UZ(-\theta_1(\hat{n}))\,\UTC\big(2\pi/\sqrt{6}\big) \,\UZ(\theta_1(\hat{n}))\,\UTC\big(\theta_2(\hat{n})-\theta_2(\hat{n}_2)\big)\\[4pt]&\qquad\times\UZ(-\theta_1(\hat{n}_2))\,\UTC\big(2\pi/\sqrt{6}\big)\,\UZ(\theta_1(\hat{n}_2))\,\UTC\big(\theta_2(\hat{n}_2)-\theta_2(\hat{n}_1)\big)\\[4pt]&\qquad\times\UZ(-\theta_1(\hat{n}_1))\,\UTC\big(2\pi/\sqrt{6}\big)\,\UZ(\theta_1(\hat{n}_1))\,\UTC\big(\theta_2(\hat{n}_1)\big)\,.
 \end{split}
\end{align}
Recall that $\UTC(r):=\exp(-ir\HTC/\gTC)$ and $\UZ(\theta):=\exp(-i\theta J_z)$.
Therefore, the total interaction time is
\begin{align}
    \tau_A(\text{3-step}) = \left[ 3\times\frac{2\pi}{\sqrt{6}} + |\theta_2(\hat{n})| + |\theta_2(\hat{n})-\theta_2(\hat{n}_2)| + |\theta_2(\hat{n}_2)-\theta_2(\hat{n}_1)| + |\theta_2(\hat{n}_1)|\right]\times\gTC^{-1}\,.
 \label{eq:Atime_3step}
\end{align}
For any axis $\hat{n}$, angle $\theta_2(\hat{n})\in\frac{1}{\sqrt{2}}[-\pi,\pi)$, which implies $|\theta_2(\hat{n})|$,  $|\theta_2(\hat{n})-\theta_2(\hat{n}_2)|$, $|\theta_2(\hat{n}_2)-\theta_2(\hat{n}_1)|$, and $|\theta_2(\hat{n}_1)|$ each lie within $\big[0,\pi/\sqrt{2}\big)$.
Therefore,
\begin{align}
    \frac{3}{\sqrt{6}}\times2\pi\gTC^{-1} \quad\leq\quad\tau_A(\text{3-step}) \quad\leq\quad \left[\frac{3}{\sqrt{6}} + \frac{4}{2\sqrt{2}}\right]\times2\pi\gTC^{-1}\,.
 \label{eq:3step_time_bounds}
\end{align}
On the other hand, the $2k$-step circuit of \cref{eq:2k_step_unitary} reads
\begin{align}
 \begin{split}
    A(U) &\eq \UTC\big(-\theta_2(\hat{n})\big)\,\UZ(-\theta_1(\hat{n})) \,\UTC\big(2k\pi/\sqrt{6}\big)\, \UZ(\theta_1(\hat{n}))\,\UTC\big(\theta_2(\hat{n})-\theta_2(\hat{m})\big) \\[4pt]
    &\qquad\times \UZ(-\theta_1(\hat{m}))\,\UTC\big(2k\pi/\sqrt{6}\big)\,\UZ(\theta_1(\hat{m})\UTC\big(\theta_2(\hat{m})\big)\,,
 \end{split}
\end{align}
for which the total interaction time is
\begin{align}
    \tau_A(2k\text{-step}) = \left[\frac{(2k)2\pi}{\sqrt{6}} + |\theta_2(\hat{n})| + |\theta_2(\hat{n})-\theta_2(\hat{m})| + |\theta_2(\hat{m})|\right]\times\gTC^{-1}\,,
 \label{eq:Atime_2step}
\end{align}
which means
\begin{align}
    \frac{2k}{\sqrt{6}}\times2\pi\gTC^{-1} \quad\leq\quad\tau_A(2k\text{-step})\quad\leq\quad \left[\frac{2k}{\sqrt{6}} + \frac{3}{2\sqrt{2}}\right]\times2\pi\gTC^{-1}\,.
 \label{eq:2step_time_bounds}
\end{align}
This yields \cref{eq:time_bound_2step} and \cref{eq:time_bound_4step} for $k=1$ and $k=2$, respectively.
\end{proof}
Note that practically, fewer steps generally means smaller interaction times, but depending on the Euler angles for a given decomposition, this is not always the case.

\appsubsec{Minimizing total interaction time of an $A$-gate}{app:Agate_optimization}
The interaction time of any particular $A$-gate (\cref{eq:Atime_3step,eq:Atime_2step}) depends on the Euler angles in \cref{eq:Rnk_gate}.
In fact, neither the 2-step decompositions of \cref{lem:2step} nor the Euler decompositions in \cref{eq:Rn_gate} are unique.
Here, we describe how to optimize over these degrees of freedom to find decompositions that minimize the total interaction time.

To implement an $A$-gate $A(U)$, where $U:=\exp(i\al\hat{\mu}\cdot\vec{\sigma})$, \cref{prop:Agate_decomp} states that we may use a 2-step, 3-step, or 4-step decomposition, depending on the value of $\al$.
Note that the ranges on $\al$ for realizing 3-step and 4-step decompositions overlap, as do the ranges on possible interaction times.
Specifically, for any $U$ with $-0.36\approx\cos(4\de)\leq\cos(\al)\leq\cos(3\de)\approx-0.11$, we should optimize over both 3-step and 4-step decompositions.
Also, note that the 3-step decomposition in \cref{eq:three_cases} is a 2-step decomposition followed by a fixed rotation about $\hat{n}$, so optimizing a 3-step decomposition boils down to optimizing a $2k$-step decomposition of the form
\begin{align}
    U := \exp(i\al\,\hat{\mu}\cdot\vec{\sigma}) \,\,=\,\, \exp(i(k\de)\,\hat{n}\cdot\vec{\sigma})\exp(i(k\de)\,\hat{m}\cdot\vec{\sigma}) \quad:\quad\cos(\al)\geq\cos(2k\de)\,,
 \label{eq:2step}
\end{align}
where $\hat{\mu}$ is any axis, $\delta=2\pi/\sqrt{3}$ is fixed, $k=1$ (2-step) or $2$ (4-step), and axes $\hat{n},\hat{m}$ satisfy the equations
\bes
\begin{align}\label{eq:cosA}
    \cos(\al) &\eq \cos^2(k\de)-\sin^2(k\de)(\hat{n}\cdot\hat{m})\\[4pt]\label{eq:sinA}
    \hat{\mu}\sin(\al) &\eq \sin(k\de)\cos(k\de)(\hat{n}+\hat{m}) + \sin^2(k\de)(\hat{n}\times\hat{m})\,.
\end{align}
\ees
Solutions for $\hat{n},\hat{m}$ are not unique, because rotations about axis $\hat{\mu}$ leave these equations invariant.
In particular, for any solution $\hat{n},\hat{m}$ of \cref{eq:cosA,eq:sinA}, the simultaneously rotated vectors
\begin{align}
    \hat{n}_{\text{rot}} &\eq \cos(\theta)\hat{n} + \sin(\theta)(\hat{n}\times\hat{\mu}) + (1-\cos(\theta))(\hat{n}\cdot\hat{\mu})\hat{\mu}\\[4pt]
    \hat{m}_{\text{rot}} &\eq \cos(\theta)\hat{m} + \sin(\theta)(\hat{m}\times\hat{\mu}) + (1-\cos(\theta))(\hat{m}\cdot\hat{\mu})\hat{\mu}
\end{align}
are also solutions, for any $\theta\in[0,2\pi)$.

The two sets of Euler angles $\{\theta_1(\hat{n}_{\text{rot}}),\theta_2(\hat{n}_{\text{rot}})\}$ and $\{\theta_1(\hat{m}_{\text{rot}}),\theta_2(\hat{m}_{\text{rot}})\}$ corresponding to the decomposition in see \cref{eq:Rnk_gate} for rotated components $\exp(i(k\de)\,\hat{n}_{\text{rot}}\cdot\sigma)$ and $\exp(i(k\de)\,\hat{m}_{\text{rot}}\cdot\sigma) $ depend on the overall rotation parameter $\theta$.
Therefore, to construct the optimal $A(U)$, we optimize the contribution from these Euler angles to the total interaction time,
\begin{align}
    |\theta_2(\hat{n}_{\text{rot}})| + |\theta_2(\hat{n}_{\text{rot}})-\theta_2(\hat{m}_{\text{rot}})| + |\theta_2(\hat{m}_{\text{rot}})|\,,
 \label{eq:Euler_angle_time}
\end{align}
as a function of $\theta\in[0,2\pi)$.

Second, for a given axis $\hat{n}$, there are four possible sets of valid Euler angles, according to \cref{eq:Rn_gate}.
In particular, if $\hat{n}$ is written in components $\hat{n}=(n_x,n_y,n_z)$, \cref{eq:Rn_gate} reads
\begin{align}
 \begin{split}
    \begin{matrix*}[l]
        n_x &\eq \cos(2\ga) \\[4pt]
    n_y &\eq -\sin(2\ga)\cos(2\beta)\\[4pt]
    n_z &\eq \sin(2\ga)\sin(2\beta)
    \end{matrix*}\hspace{20pt}
    \begin{cases}
        \theta_2(\hat{n}):=\frac{\beta}{\sqrt{2}}\\[6pt]\theta_1(\hat{n}):=2\ga
    \end{cases}
 \end{split}
 \label{eq:Euler_angles}
\end{align}
Thus, if $(\ga,\beta)$ is a solution, then $(\ga,\beta+\pi)$, $(-\ga,\beta+\pi/2)$, $(-\ga,\beta-\pi/2)$ are as well.
Therefore, for each set of axes $(\hat{n}_{\text{rot}},\hat{m}_{\text{rot}})$, we optimize \cref{eq:Euler_angle_time} over these four sets of Euler angles.

For a 3-step decomposition,
\begin{align}
    \exp\big(i \de\,[\pm\hat{n}]\cdot \vec{\sigma}\big)  \exp\big(i \de\,\hat{n}_2\cdot \vec{\sigma}\big)\exp\big(i \de\,\hat{n}_1\cdot \vec{\sigma}\big)\,,
\end{align}
we follow the same procedure described above, except that the quantity to be minimized is
\begin{align}
    |\theta_2(\hat{n})|+|\theta_2(\hat{n})-\theta_2(\hat{n}_2)| + |\theta_2(\hat{n}_2)-\theta_2(\hat{n}_1)| + |\theta_2(\hat{n}_1)|\,.
 \label{eq:Euler_angle_time_3step}
\end{align}

\appsubsec{Minimizing interaction times for any 2-qubit, PI, $\UU(1)$-invariant unitary}{app:minimizing_2qubit_unis}
In \cref{sec:useful_2q_gate}, we showed that any two-qubit PI unitary which also conserves $J_z$ can be implemented as 
\begin{align}
    U(\theta,\theta_+,\theta',\al)= e^{i\al}R_z(\theta')A(\theta,\theta_+)FR_z(\theta)F\,.
 \label{eq:F_combo2}
\end{align}
In fact, the job of the $F$-gate -- swapping the states $\ket{00}\otimes\vac$ and $\ket{11}\otimes\ket{2}_{\text{osc}}$ -- is done equally well by $F^{\dag}$.
Additionally, in the $q=2,j=1$ sector, $F$ is Hermitian and thus is its own inverse, $\pi_{2,1}(F)=\pi_{2,1}(F^{\dag})=\pi_{2,1}(F^{-1})$.
However, $F$ and $F^{\dag}$ act differently in the $q=1,j=1$ sector.
This implies that the $A$-gate in \cref{eq:F_combo2} may be different, depending on which combination of $F$ and $F^{\dag}$ are used.
In particular, for a given unitary $U\in\H_{1,1}\simeq\SU(2)$, we should optimize for the fastest $A$-gate over the four possible combinations of $F$ and $F^{\dag}$, i.e.
\bes
\begin{align}
    U &\eq e^{i\al}R_z(\theta')A_1(\theta,\theta_+)FR_z(\theta)F \\[2pt]
    U &\eq e^{i\al}R_z(\theta')A_2(\theta,\theta_+)F^{\dag}R_z(\theta)F \\[2pt]
    U &\eq e^{i\al}R_z(\theta')A_3(\theta,\theta_+)FR_z(\theta)F^{\dag} \\[2pt]
    U &\eq e^{i\al}R_z(\theta')A_4(\theta,\theta_+)F^{\dag}R_z(\theta)F^{\dag}\,,
\end{align}
\ees
where $A_{\ell}(\theta,\theta_+)$ depends on which combination is used.
Therefore, to find the $A$-gate with minimal interaction time, optimize each $A_{\ell}(\theta,\theta_+):=\exp(i\al_{\ell}\,\mu_{\ell}\cdot\vec{\sigma})$, for $\ell=1,2,3,4$, as described in the previous section.

\appsubsec{Examples: CZ, SWAP, iSWAP, $\sqrt{\text{iSWAP}}$}{app:useful_gate_table}
For these four useful gates, \cref{tab:gates_info} gives the optimal interaction times using our constructions, as well as the corresponding parameters in terms of the decomposition
\begin{align*}
    U = e^{i\al}R_z(\theta')A(\theta,\theta_+)FR_z(\theta)F\,,
\end{align*}
for those four gates, as well as $U_{ZZ}(\phi):=\exp(-i\phi Z_1Z_2)$ and $U_{\Psi^+}(-2\pi/\sqrt{3}):=\exp(i\pure{\Psi^+}2\pi/\sqrt{3})$,
\begin{table}[htp]
    \centering
    \renewcommand{\arraystretch}{1.5}
    \setlength{\tabcolsep}{6pt}
    \begin{tabular}{l!{\vrule width1.5pt}l!{\vrule width1.5pt}l|l|l|l|l}
        \textbf{Gate} &  \shortstack{\textbf{Time} \\
        ($\tau\times\gTC/2\pi$)} & $\al$ & $\theta'$ & $\theta$ & $\theta_+$ & \shortstack{Full Gate \\ Sequence}\\\hline
        CZ & $\approx2.866$ & $0$ & $\pi/2$ & $\pi/2$ & $0$ & $R_z(\frac{\pi}{2})A(\frac{\pi}{2},0)FR_z(\frac{\pi}{2})F^{\dag}$ \\
        SWAP & $\approx1.273$ & $\pi$ & $0$ & $\pi$ & $\pi$ & $e^{i\pi}A(\pi,\pi)R_z(\pi)$ \\
        iSWAP & $\approx2.546$ & $-\pi/2$ & $0$ & $\pi/2$ & $\pi$ & $e^{-i\pi/2}A(\frac{\pi}{2},\pi)FR_z(\frac{\pi}{2})F^{\dag}$ \\
        $\sqrt{\text{iSWAP}}$ & $\approx2.688$ & $-\pi/4$ & $0$ & $\pi/4$ & $\pi/2$ & $e^{-i\pi/4}A(\frac{\pi}{4},\frac{\pi}{2})F^{\dag}R_z(\frac{\pi}{4})F^{\dag}$\\
        $U_{ZZ}(\phi)$ & $\lesssim3.92$ & $\phi$ & $0$ & $-2\phi$ & $0$ & $e^{i\phi}A(-2\phi,0)FR_z(-2\phi)F$ \\
        $U_{\Psi^+}\Big($$-\frac{2\pi}{\sqrt{3}}$$\Big)$ & $\approx 0.585$ & $0$ & $0$ & $0$ & $\frac{2\pi}{\sqrt{3}}$ & $\UTC(-\frac{\pi}{4\sqrt{2}}) \UZ(-\frac{\pi}{2}) \UTC(\frac{2\pi}{\sqrt{6}}) \UZ(\frac{\pi}{2}) \UTC(\frac{\pi}{4\sqrt{2}})$
    \end{tabular}
    \caption{Implementing five two-qubit gates: interaction times, and decompositions in terms of \cref{eq:2qubit_decomp}.}
    \label{tab:gates_info}
\end{table}

\appsubsec{Parameters of $A$-gates used to implement CZ, SWAP, iSWAP, $\sqrt{\text{iSWAP}}$, $\Vswap$}{app:A_gate_params}
\Cref{tab:A_gate_params} lists the parameters of different $A$-gates used to implement CZ, SWAP, iSWAP, $\sqrt{\text{iSWAP}}$, $\Vswap$, and $A(\pi_{1,1}(F^{\dag})$.
For each $A$-gate, the four parameters $\theta_1(\hat{n}),\theta_2(\hat{n})$ and $\theta_1(\hat{m}),\theta_2(\hat{m})$ are precisely the Euler angles from \cref{eq:Rn_gate} used to implement $R_{\hat{n}}(4\pi/\sqrt{3})$ and $R_{\hat{m}}(4\pi/\sqrt{3})$, respectively.
We have checked by explicit matrix multiplication that the circuits in \cref{tab:circuit_table,fig:swap_circuit} implement the desired unitaries.
\begin{table}[htp]
 \begin{adjustbox}{width=1.0\textwidth}
    \centering
    \renewcommand{\arraystretch}{1.5}
    \setlength{\tabcolsep}{6pt}
    \begin{tabular}{l|l|l|l|l|l|l|l|l}
     & \shortstack{Decomp.\\Type} & \shortstack{$A$-gate Time \\$(\tau\times\gTC/2\pi)$} & $\theta_1(\hat{n}_2)\approx$ & $\theta_2(\hat{n}_2)\approx$ & $\theta_1(\hat{n}_1)\approx$ & $\theta_2(\hat{n}_1)\approx$ & $\theta_1(\hat{n})\approx$ & $\theta_2(\hat{n})\approx$ \\\hline
     CZ & 4-step & $\approx1.641$ & $-0.12147118$ & $-0.02458664$ & $-1.74926322$ & $-0.02483674$ && \\
     SWAP & 3-step & $\approx1.273$ & $2.62219567$& $0.05025862$ & $0.56293308$ & $0.04976850$ & $1.59079766$ & $-0.10034647$ \\
     iSWAP & 3-step & $\approx1.321$ & $-2.45780412$ & $-1.12604357$ & $-1.21165852$ & $-1.41005151$ & $1.15257779$ & $0.80892003$ \\
     $\sqrt{\text{iSWAP}}$ & 3-step & $\approx1.464$ & $2.28652854$ & $0.58015043$ & $1.69646350$ & $-0.09519097$ & $2.54885204$ & $0.07041776$ \\
     $\Vswap$ & 2-step & $\approx0.829$ & $0.77219211$& $-0.03972316$ & $2.36943563$ & $-0.03968950$ && \\
    $A(\pi_{1,1}(F^{\dag}))$ & 2-step & $\approx0.833$ & $1.65858418$ & $0.05233908$ & $1.09721596$ & $0.05232010$ &&
    \end{tabular}
 \end{adjustbox}
    \caption{Parameters for optimal $A$-gates used for the circuits in \cref{tab:circuit_table,fig:swap_circuit}, as well as $A(\pi_{1,1}(F^{\dag}))$. The six angles correspond to the decompositions in \cref{eq:three_cases} for $2$, $3$, or $4$-step decompositions, and they are obtained from numerical solutions to \cref{eq:cosA} and \cref{eq:Euler_angles}.}
    \label{tab:A_gate_params}
\end{table}

\appsubsec{$F$-gate}{app:F_gate}
In this section, we construct the $F$-gate, which implements the map in \cref{eq:F_map} and also acts non-trivially on the 2D subspace in \cref{eq:2D}.
First, consider the $q=2$, $j=1$ sector $\H_{2,1}$, with the ordered basis $\big(\ket{00}\otimes\vac,\,\ket{\Psi^+}\otimes\ket{1}_{\text{osc}},\,\ket{11}\otimes\ket{2}_{\text{osc}}\big)$. 
The eigenvalues of $\HTC$ in this sector are $0$ and $\pm\sqrt{6}$, so its time evolution operator has period $2\pi/\sqrt{6}$.
Consider the time evolution corresponding to one half of a full period,
\begin{align}
    V\,\,:=\,\,\pi_{2,1}\Big(\UTC\big(\pi/\sqrt{6}\big)\Big) \,\,:=\,\,\pi_{2,1}\left( e^{-i\HTC\pi/\sqrt{6}}\right) \eq \left(\begin{matrix} \frac{1}{3} & 0 & -\sqrt{\frac{8}{9}} \\ 0 & -1 & 0 \\ -\sqrt{\frac{8}{9}} & 0 & -\frac{1}{3}\end{matrix}\right)\,.
 \label{eq:Vdef}
\end{align}
We prove that by combining three copies of the above $3\times 3$ unitary with proper  $z$-rotations, one can realize 
\begin{align}
    \pi_{2,1}\Big(\UZ\left(\phi_0/2\right)\,V\,\UZ\left(-\phi_1\right)\,V\, \UZ\left(\phi_1\right)\,V\,\UZ\left(-\phi_0/2\right)\Big) \eq \left(\begin{matrix}0&0&1\\0&-1&0\\1&0&0\end{matrix}\right)\,,
 \label{eq:F_gate_matrix}
\end{align}
which implements the map
\begin{align}
 \begin{split}
    \ket{00}\otimes\vac&\longrightarrow\,\,\,\,\,\ket{11}\otimes\ket{2}_{\text{osc}}\\
   \ket{\Psi^+}\otimes\ket{1}_{\text{osc}}&\longrightarrow-\ket{\Psi^+}\otimes\ket{1}_{\text{osc}}\\
    \ket{11}\otimes\ket{2}_{\text{osc}}&\longrightarrow\,\,\,\,\,\,\ket{00}\otimes\vac\,.
 \end{split}
 \label{eq:F_map2}
\end{align}
Explicitly, it is straightforward to check that for angles
\begin{align}
    \phi_1 \,:=\, \frac{1}{2}\arccos\left(\frac{7}{16}\right)\qquad\text{and}\qquad\phi_0 \,:=\, \arctan\left(-\frac{\sqrt{23}}{3}\right) + \pi\,,
 \label{eq:phi0_phi1}
\end{align}
the above equation is satisfied.
At the end of this section, we give an alternative proof using a geometric argument that explains why there should exist angles $\phi_0$ and $\phi_1$ satisfying the above identity.

Therefore, we define
\begin{align}
    F \,:=\, \UZ\left(\phi_0/2\right)\,\UTC\big(\pi/\sqrt{6}\big)\,\UZ\left(-\phi_1\right)\,\UTC\big(\pi/\sqrt{6}\big)\, \UZ\left(\phi_1\right)\,\UTC\big(\pi/\sqrt{6}\big)\,\UZ\left(-\phi_0/2\right)\,.
\end{align}
Recall that $F$ also acts non-trivially on the 2D subspace in \cref{eq:2D}, i.e., the $q=1$, $j=1$ sector $\H_{1,1}$.
Recall from \cref{eq:charge1_pauli} that if $\big(\ket{\Psi^+}\otimes\vac,\ket{11}
\otimes\ket{1}_{\text{vac}}\big)$ is an ordered basis of $\H_{1,1}$, then $\HTC$ and $\Zhat$ act like $\sigma_x$ and $\sigma_z$ in sector $\H_{1,1}$.
Then, the component of $F$ in this sector is
\begin{align}
 \label{eq:Fgate_q=1_relic}
     \pi_{1,1}(F) &= e^{-i\sigma_z\phi_0/4}\,e^{-i\sigma_x\pi/\sqrt{3}}\,e^{i\sigma_z\phi_1/2}\,e^{-i\sigma_x\pi/\sqrt{3}}\,e^{-i\sigma_z\phi_1/2}\,e^{-i\sigma_x\pi/\sqrt{3}} e^{i\sigma_z\phi_0/4} \\[10pt]\nonumber
     &\hspace{-16pt}=\begin{pmatrix}
         \cos\left(\frac{\de}{2}\right)\Big[\big(1+\cos(\phi_1)\big)\cos(\de)-\cos(\phi_1)\Big] & -i\sin\left(\frac{\de}{2}\right)e^{-i\frac{\phi_0}{2}}\Big[\big(1+\cos(\phi_1)\big)\cos(\de)+i\sin(\phi_1)+1\Big]\\
         -i\sin\left(\frac{\de}{2}\right)e^{i\frac{\phi_0}{2}}\Big[\big(1+\cos(\phi_1)\big)\cos(\de)-i\sin(\phi_1)+1\Big]& \cos\left(\frac{\de}{2}\right)\Big[\big(1+\cos(\phi_1)\big)\cos(\de)-\cos(\phi_1)\Big]
     \end{pmatrix}\,,
\end{align}
with $\de:=\frac{2\pi}{\sqrt{3}}$ and $\phi_0,\phi_1$ as in \cref{eq:phi0_phi1}.
Note that $\pi_{1,1}(F)$ is a Bloch sphere rotation about an axis in the $x$-$y$ plane.

The total interaction time of an $F$-gate is
\begin{align}
    \tau_F = \frac{3\pi}{\sqrt{6}} \,\,\approx\,\,0.612\times2\pi\gTC^{-1}\,,
\end{align}
and the corresponding circuit is illustrated in \cref{tab:circuit_table}.

If one wishes to realize the map in \cref{eq:F_map}, i.e. swapping states $\ket{00}\otimes\vac\longleftrightarrow\ket{11}\otimes\ket{2}_{\text{osc}}$, \textit{without} a non-trivial relic in the 2D sector $\H_{1,1}$, then apply an $A$-gate which cancels out $\pi_{1,1}(F)$, that is $A(\pi_{1,1}(F^{\dag}))$.
The interaction time of this $A$-gate is approximately $0.833\times2\pi\gTC^{-1}$ (see \cref{tab:A_gate_params}).

\appsubsec{A geometric explanation of the $F$-gate}{app:F_gate_geometry}
Here, we present a geometric argument that explains why there exists  $\phi_0$ and $\phi_1$ satisfying  \cref{eq:F_gate_matrix}.
Consider the two-level subspace spanned by $\ket{0}_{\text{eff}}:=\ket{00}\otimes\vac$ and $\ket{1}_{\text{eff}}:=\ket{11}\otimes\ket{2}_{\text{osc}}$, and let these states denote the $\pm1$ eigenstates, respectively, of Pauli matrix $\sigma_z$.
Then, in this subspace, $V$ (defined in \cref{eq:Vdef}) and $\UZ(\theta)$ act as
\begin{align}
 \label{eq:2dim_sector_V}
    \widetilde{V} \,=&\,\,e^{i\pi/2}\exp(-i\sigma_z\pi/2)\exp\big(-i\sigma_y\theta_0/2\big)\,\,,\quad\theta_0=2\arccos(1/3)\approx0.78\pi\,. \\[4pt]
    \widetilde{R}_z(\theta) \,=&\,\, \exp\big(-i\sigma_z\theta\big)
 \label{eq:2dim_sector_Rz}
\end{align}
In other words, on the Bloch sphere, $\widetilde{R}_z(\theta)$ is a rotation by $2\theta$ about $z$, and $\widetilde{V}$ is, up to a global phase, a rotation by fixed angle $2\arccos(1/3)\approx0.78\pi$ about $y$, followed by a $\pi$-rotation about $z$.
Also, note that unitary $\widetilde{V}$ is also Hermitian, and hence its own inverse.

To obtain the map in \cref{eq:F_map}, we argue geometrically that there are angles $\phi_0,\phi_1$ such that
\begin{align}
    \widetilde{V}\widetilde{R}_z(-\phi_1)\widetilde{V}\widetilde{R}_z(\phi_1)\widetilde{V}\, \ket{0}_{\text{eff}} \,\eq\, e^{-i\phi_0} \ket{1}_{\text{eff}}\,,
\end{align}
or equivalently,
\begin{align}
    \widetilde{R}_z(-\phi_1)\widetilde{V}\widetilde{R}_z(\phi_1)\times\big(\widetilde{V}\ket{0}_{\text{eff}}\big) \,\eq\, e^{-i\phi_0} \big(\widetilde{V}\ket{1}_{\text{eff}}\big)\,.
 \label{eq:F_geometry_map}
\end{align}
Using \cref{eq:2dim_sector_V,eq:2dim_sector_Rz}, and bringing the $\exp(-i\sigma_z\pi/2)$ from inside $\widetilde{V}$ on the left-hand side over to the right-hand side, \cref{eq:F_geometry_map} can be written as
\begin{align}
    \exp\big(i\sigma_z\phi_1\big)\exp\big(-i\sigma_y\theta_0/2\big) \exp\big(-i\sigma_z\phi_1\big) \widetilde{V}\ket{0}_{\text{eff}} \eq e^{-i(\phi_0+\pi/2)}\exp\big(i\sigma_z\pi/2\big)\widetilde{V}\ket{1}_{\text{eff}}\,.
 \label{eq:F_geometry2}
\end{align}
For an appropriate choice of $\phi_1\in[0,2\pi)$, the operator
\begin{align}
    \exp\big(i\sigma_z\phi_1\big)\exp\big(-i\sigma_y\theta_0/2\big) \exp\big(-i\sigma_z\phi_1\big) \,\eq\, \exp\Big(-i(\theta_0/2)\big(\sin(\phi_1)\hat{x}+\cos(\phi_1)\hat{y}\big)\cdot\vec{\sigma}\Big)
 \label{eq:xy_rotation}
\end{align}
realizes a rotation by fixed angle $\theta_0$ about an \textit{arbitrary} axis in the $x$-$y$ plane.

The vectors $\exp\big(i\sigma_z\pi/2\big)\widetilde{V}\ket{1}_{\text{eff}}$ and $\widetilde{V}\ket{0}_{\text{eff}}$ from \cref{eq:F_geometry2} are illustrated in purple in \cref{fig:F_gate_geometry}.
Importantly, these vectors lie in the $x$-$z$ plane, and 
because the angle $2\theta_0-\pi\approx0.57\pi$ between them is small enough, i.e. less than fixed angle $\theta_0\approx0.78\pi$, a suitable rotation of the form in \cref{eq:xy_rotation} can map one vector to the other; up to a potential phase, which is denoted $-(\phi_0+\pi/2)$.
Therefore, the sequence $\widetilde{V}\widetilde{R}_z(-\phi_1)\widetilde{V}\widetilde{R}_z(\phi_1)\widetilde{V}$ maps $\ket{0}_{\text{eff}}$ to $e^{-i\phi_0}\ket{1}_{\text{eff}}$.
Finally, conjugating this sequence by $\widetilde{R}_z(\phi_0/2)$ removes undesired phases and yields \cref{eq:F_gate_matrix}.

\begin{figure*}[htp]
    \centering
    \includegraphics[width=0.9\textwidth]{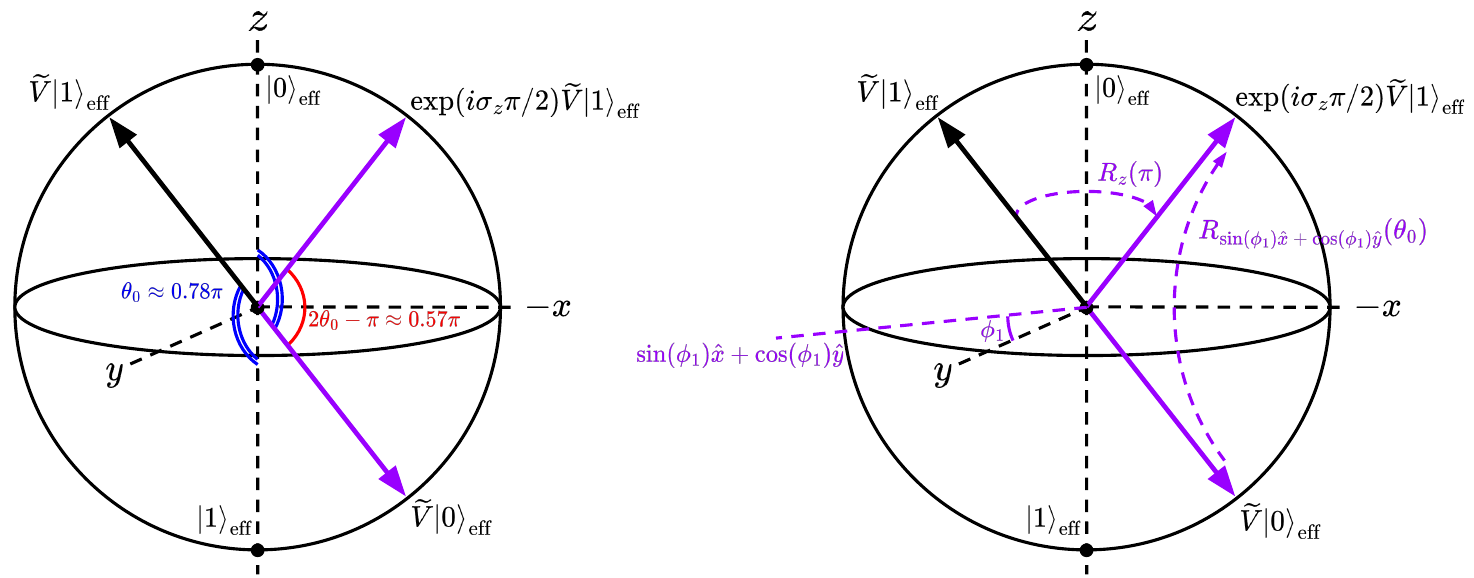}
    \caption{\textbf{Bloch Sphere geometry of the $F$-gate:} The left picture illustrates the locations of $\widetilde{V}\ket{0}_{\text{eff}}$, $\widetilde{V}\ket{1}_{\text{eff}}$, and  $\exp(i\sigma_z\pi/2)\widetilde{V}\ket{1}_{\text{eff}}$ on the Bloch sphere. 
    All three vectors lie in the $x$-$z$ plane.
    The right picture illustrates how $\widetilde{V}\ket{0}_{\text{eff}}$ is mapped to $\exp(i\sigma_z\pi/2)\widetilde{V}\ket{1}_{\text{eff}}$ via the rotation in \cref{eq:xy_rotation}.}
    \label{fig:F_gate_geometry}
\end{figure*}

\newpage
\appsec{Qubit-Oscillator SWAP for $n=2$  qubits}{app:state_prep_circuit}
In this section, we construct the unitary $\Vswap$ for the case of $n=2$ qubits, which implements the gadget described in \cref{sec:oscillator_state_prep}.
\begin{align}
    \Vswap\big(\ket{\psi}\otimes\vac\big) \eq \ket{1}^{\otimes 2}\otimes\ket{\Psi}_{\text{osc}}\,,
\end{align}
where $\ket{\psi}\in\C^{3}$ is an arbitrary state in the 3-dimensional symmetric (triplet) space of two qubits.

As a matrix in the ordered  basis from \cref{eq:2q_subspace}, $\Vswap$ can be implemented by the following unitary:
\begin{align*}
    \Vswap \eq \arraycolsep=4pt\def\arraystretch{1.4}\left(
    \begin{array}{c|cc|ccc}
        1 &&&&&\\ \hline
        & 0 & -1 &&& \\
        & 1 & 0 &&& \\ \hline
        &&& 0 & 0 & 1 \\
        &&& 0 & -1 & 0 \\
        &&& 1 & 0 & 0
    \end{array}
    \right)\,.
\end{align*}
Note that $\ket{\psi}\otimes\vac$ only has nonzero components in the lowest three charge sectors, so it is sufficient to consider just these sectors.
First, the $F$-gate introduced in \cref{sec:useful_2q_gate}, and constructed in \cref{app:F_gate}, implements the desired component of $\Vswap$ in the 3D $q=2$ sector. 
Then, note that in each of the $q=0$ (1D) and $q=1$ (2D) sectors, this unitary has determinant one, so the components of $F^{\dag}$ in these sectors are
\begin{align*}
    F^{\dag} 
    \eq \arraycolsep=4pt\def\arraystretch{1.4}\left(
    \begin{array}{c|c|ccc}
        1 &&&&\\ \hline
        & \widetilde{V}_1 &&& \\ \hline
        && 0 & 0 & 1 \\
        && 0 & -1 & 0 \\
        && 1 & 0 & 0
    \end{array}
    \right)\,,
\end{align*}
where $\widetilde{V}_1\in\SU(2)$ is given by the fixed unitary in \cref{eq:Fgate_q=1_relic}.
Second, apply the $A$-gate
\begin{align}
    A_{\text{swap}} \,\,:=A\left((-i\sigma_y)\widetilde{V}_1^{\dag}\right) \eq \arraycolsep=6pt\def\arraystretch{1.8}\left(
    \begin{array}{c|c|c}
        1 && \\ \hline
        & (-i\sigma_y)\widetilde{V}_1^{\dag} & \\ \hline 
        && \1
    \end{array}
    \right)\,,
 \label{eq:Swap_A1}
\end{align}
to implement the desired component in the $q=1$ sector, such that
\begin{align}
 \begin{split}
    \Vswap &\eq A_{\text{swap}}\cdot F^{\dag} \eq \arraycolsep=4pt\def\arraystretch{1.4}\left(
    \begin{array}{c|cc|ccc}
        1 &&&&&\\ \hline
        & 0 & -1 &&& \\
        & 1 & 0 &&& \\ \hline
        &&& 0 & 0 & 1 \\
        &&& 0 & -1 & 0 \\
        &&& 1 & 0 & 0
    \end{array}
    \right)\,.
 \end{split}
\end{align}
This particular $A$-gate is 2-step (see details in \cref{app:A_gate}, circuit in \cref{tab:gadget_circuits}), and its interaction time is
\begin{align}
    \tau_{\text{swap}} \eq \tau_F + \tau_{A_{\text{swap}}} \,\,\approx\,\,1.44\times2\pi\gTC^{-1}\,.
\end{align}

\end{document}